\documentclass[11pt,letterpaper]{article}

\usepackage{amsmath,amsthm,amsfonts,amssymb}
\usepackage{thm-restate}
\usepackage{fullpage}
\usepackage[utf8]{inputenc}
\usepackage[dvipsnames]{xcolor}
\usepackage{xspace,enumerate}
\usepackage[shortlabels]{enumitem}
\usepackage[hypertexnames=false,colorlinks=true,urlcolor=Blue,citecolor=Green,linkcolor=BrickRed]{hyperref}
\usepackage[capitalise]{cleveref}
\usepackage[OT4]{fontenc}
\usepackage[ruled]{algorithm}
\usepackage[noend]{algorithmic}
\usepackage{ifpdf}
\usepackage{todonotes}
\usepackage{enumerate}
\usepackage{thmtools}
\usepackage[mathlines]{lineno}

\usepackage[margin=1in]{geometry}
\def\poly{\operatorname{poly}}
\def\polylog{\operatorname{polylog}}

\title{Sensitivity and Dynamic Distance Oracles via Generic Matrices\\and Frobenius Form}

\newcommand{\email}[1]{\href{mailto:#1}{#1}}

\date{\vspace{-5ex}}
\author{Adam Karczmarz\thanks{University of Warsaw and IDEAS NCBR, Poland. \email{a.karczmarz@mimuw.edu.pl}. Partially supported by the ERC CoG grant TUgbOAT no 772346 and the National Science Centre (NCN) grant no. 2022/47/D/ST6/02184.}
\and Piotr Sankowski\thanks{University of Warsaw, IDEAS NCBR, and MIM Solutions, Poland. \email{sank@mimuw.edu.pl}. Partially supported by the ERC CoG grant TUgbOAT no 772346 and the National Science Centre (NCN) grant no. 2020/37/B/ST6/04179.}}

\newcommand{\Ot}{\ensuremath{\widetilde{O}}}
\newcommand{\eps}{\ensuremath{\epsilon}}
\newcommand{\dist}{\delta}

\newcommand{\wei}{w}
\newcommand{\field}{\mathbb{F}}

\theoremstyle{plain}
\newtheorem{theorem}{Theorem}[section]
\newtheorem{lemma}[theorem]{Lemma}
\newtheorem{corollary}[theorem]{Corollary}

\newtheorem{fact}[theorem]{Fact}

\begin{document}

\maketitle

\begin{abstract}
  Algebraic techniques have had an important impact on graph algorithms so far. Porting
  them, e.g., the matrix inverse, into the dynamic regime improved best-known bounds for various dynamic graph problems.
  In this paper, we develop new algorithms for another cornerstone algebraic primitive, the Frobenius normal form (FNF).
  We apply our developments to dynamic and fault-tolerant exact distance oracle problems on directed graphs.

  For generic matrices $A$ over a finite field  accompanied by an FNF, we show
  (1) an efficient data structure for querying submatrices of the first $k\geq 1$ powers of $A$, and
  (2) a near-optimal algorithm updating the FNF explicitly under rank-1 updates.

  By representing an unweighted digraph using a generic matrix over a sufficiently large field (obtained by random sampling) and
  leveraging the developed FNF toolbox, we obtain:
  \begin{itemize}
    \item a conditionally optimal distance sensitivity oracle (DSO) in the case of single-edge or single-vertex failures,
    providing a partial answer to the open question of ~\cite{GuR21},
    \item a multiple-failures DSO improving upon the state of the art~\cite{BrandS19} wrt. both preprocessing and query time,
    \item improved dynamic distance oracles in the case of single-edge updates,
    \item  a dynamic distance oracle supporting vertex updates, i.e., changing all edges incident to a single vertex, in
    $\Ot(n^2)$ worst-case time and distance queries in $\Ot(n)$ time.
  \end{itemize}

\end{abstract}

\section{Introduction}
 Algebraic techniques have had an important impact on graph algorithms so far.
 For example, the state-of-the-art maximum matching algorithm in dense non-bipartite graphs~\cite{Harvey09, MuchaS04} is
 of algebraic nature.
 Porting some fundamental linear-algebraic concepts, like the matrix inverse, into the dynamic regime
 has led to non-trivial dynamic algorithms for multiple graph problems,
 such as reachability, shortest paths, maximum matchings~\cite{BrandNS19, Sankowski04, Sankowski05, Sankowski07},
 and a unified view on iterative optimization methods~\cite{Brand21}.

  In this paper, we consider another cornerstone algebraic primitive, the Frobenius normal form (FNF).
  Any square matrix $A$ over a field $\field$ is similar to a block-diagonal matrix \linebreak
  $F=\operatorname{diag}(C_{f_1},\ldots,C_{f_k})$
  where each $C_{f_i}$ is the companion matrix of a certain monic polynomial $f_i\in\field[x]$ called an \emph{invariant factor} of~$A$.
  Such a matrix,  along with the corresponding similarity transform and its inverse, constitutes a Frobenius normal form of $A$.
  Similarly to the Jordan normal form, FNF encodes the characteristic polynomial
  of a matrix; however, contrary to the Jordan form, computing it does not require finding zeros of this polynomial.
  Computing an FNF is a well-studied problem in the symbolic computation
  community, e.g.,~\cite{Giesbrecht95, Storjohann01}.
  Using FNF, \cite{SankowskiW19} reproduced\footnote{The paper \cite{SankowskiW19} contains a rather significant error, as confirmed by its authors (personal communication). \cite{SankowskiW19} mistakenly state that the diameter $D$ (i.e., the largest finite distance) of a digraph $G$ is bounded by the degree of the smallest invariant factor of its adjacency matrix $A$.
Indeed, $D$ is bounded by the degree of the \emph{minimal polynomial} of $A$, which equals the \emph{largest} (and not the smallest) invariant factor of $A$.
Nevertheless, the construction of~\cite{SankowskiW19} remains correct if it happens that the adjacency matrix of $A$ has a single invariant factor.
}
 the Yuster-Zwick exact distance oracle's bounds~\cite{YusterZ05} (up to polylogarithmic factors),
  thus providing a proof of concept of the usage of the Frobenius form in graph data structures.
  However, it seems that the full potential of FNF and other matrix forms
  in graph algorithms is yet to be fully uncovered.

  In this paper, we develop new data structures maintaining and exploiting the Frobenius normal form
  of a matrix in the \emph{generic} case.
  The genericity assumption is a common one in the computer algebra community
  and says, broadly speaking, that an algorithm should work ``almost always'', or, for ``all but special cases''.
  A specific genericity assumption (that we also use; see, e.g.,~\cite{JeannerodV05}) for matrices may require
  that the cells of a matrix, seen as indeterminates, do not satisfy some fixed polynomial equation,
  or in other words, do not lie on some fixed hypersurface of $\field^{n\times n}$.
  In particular, such an assumption can be used to ensure that $A$ has a \emph{single invariant factor},
  or, equivalently, that the \emph{characteristic polynomial} $p_A$ of $A$ equals the \emph{minimal polynomial} $\mu_A$ of $A$.
  This property implies that $A$ is similar to the companion matrix of the characteristic polynomial of $A$, i.e., an FNF of $A$ has a particularly simple form.
  Single-invariant-factor matrices can be themselves considered \emph{generic} among all matrices over a finite field $\field$ since the fraction of matrices in $\field^{n\times n}$ not having this property is known to be $1/q^3+O\left(1/q^4\right)$ if $q$ is the field size~\cite{NeumannP95}.
  In the following, when talking about generic matrices, we mean matrices with a single invariant factor.

  By representing digraphs with generic matrices, we can apply our FNF developments to
  dynamic and fault-tolerant \emph{exact distance oracle} problems on general dense directed graphs.

\paragraph{Distance oracles in static, fault-tolerant, and dynamic settings.} In the \emph{distance oracle} problem, the goal is to preprocess the input graph $G$ into a data structure
supporting arbitrary-pair distance queries. The distance oracle problem has two trivial solutions.
First, one could precompute answers to all the $O(n^2)$ possible queries by solving the all-pairs shortest paths (APSP) problem.
The other extreme is to not preprocess the graph $G$ at all, and run an $s,t$-shortest path algorithm (such as Dijkstra's algorithm)
from scratch upon a distance query~$(s,t)$.
The study of distance oracles concentrates on identifying what \emph{non-trivial} tradeoffs between space, preprocessing time and query time
are attainable, possibly under additional assumptions about the graph class of interest, and whether approximate
answers are acceptable.

Real-world networks are subject to link/node failures and evolve in time and thus motivate the study of distance oracles in fault-tolerant
and dynamic settings.

In the \emph{distance sensitivity oracle} (DSO) problem, the goal is to preprocess the input graph $G=(V,E)$
so that queries $(s,t,F)$ asking for the length of
the shortest $s,t$ path not going through the subset  $F\subseteq V\cup E$
of \emph{failed} edges or vertices are supported.
A DSO may also constrain the number of allowed failures, e.g., require that only
a single edge or vertex fails. If only at most $k$ failures are supported, we call
such a DSO a \emph{$k$-DSO}.

In dynamic scenarios,
the input graph $G$ is subject to edge set updates and we seek
a data structure supporting distance queries interleaved with graph updates.
In the \emph{fully dynamic} setting, the data structure is supposed to accept both
edge insertions and edge deletions. In the \emph{incremental} (\emph{decremental}, resp.) setting,
only edge insertions (deletions, resp.) are accepted.
Some dynamic distance oracles accept \emph{single-edge} updates, whereas other allow \emph{vertex updates}, i.e., changing all (possibly $\Theta(n)$) edges incident to a single vertex at once.

\subsection{State of the art}
\paragraph{Static and dynamic computation of Frobenius normal form.}
For generic matrices (as defined before), finding an FNF is closely related to computing the characteristic polynomial.
\cite{Keller-Gehrig85} showed an $\Ot(n^\omega)$
time\footnote{Where $\omega\approx 2.37$ is the matrix multiplication exponent, i.e., a number such that one can multiply
two $n\times n$ matrices in $O(n^\omega)$ time.}
 algorithm computing the characteristic
polynomial. \cite{Giesbrecht95} was the first to obtain an $\Ot(n^\omega)$-time
algorithm computing an FNF in the general (non-generic) case, whereas~\cite{Storjohann01} gave a deterministic
algorithm running within that near-optimal bound. Computing the Frobenius
form has also been studied for sparse matrices and, more generally, in the ``black box'' model
where the input matrix can only be accessed via multiplying it by vectors~\cite{eberly00, villard00}.
In particular, a Frobenius form of a generic matrix (with one invariant factor) can
be computed in $\Ot(n^2)$ time plus the time needed to perform $\Ot(n)$ black-box matrix-vector multiplications~\cite{eberly00}.

\cite{FrandsenS11} studied dynamic maintenance of an FNF of a matrix subject
to rank-1 updates (i.e., updates of the form $A:=A+ab^T$ for given vectors $a,b\in\field^{n\times 1}$, capturing, e.g., row and column updates).
They gave a dynamic algorithm with $\Ot(kn^2)$ worst-case update time, where
$k$ is the number of invariant factors of $A$. For the generic case $k=1$,
the update time is $\Ot(n^2)$. Their algorithm has a significant limitation
though. Even for generic matrices, whereas the block-diagonal matrix $F$ (encoding the characteristic polynomial)
similar to $A$ is maintained explicitly, an appropriate similarity transform $Q$ and its inverse
such that $A=Q^{-1}\cdot F\cdot Q$
is maintained \emph{only implicitly}. More specifically, the matrices $Q$ and $Q^{-1}$
can only be accessed by multiplying them via vectors in $\Ot(n^2)$ time
which makes processing them troublesome.

\paragraph{Static distance oracles.} For general weighted digraphs with large edge weights (say, integral and polynomial in $n$),
no non-trivial preprocessing/space/query trade-offs are known.

For digraphs with small integer  weights $\{-W,\ldots,W\}$, \cite{YusterZ05} gave a non-trivial distance oracle with $\Ot(Wn^\omega)$ preprocessing time, $O(n^2)$ space and
$\Ot(n)$ query time.
The query time is also significantly smaller than the $\Theta(n^2)$ cost of running breadth-first search on $G$.
The data structure of~\cite{YusterZ05} can also produce an actual shortest path (not just the distance) upon query.
Importantly, in the fundamental case of dense unweighted graphs, preprocessing time of~\cite{YusterZ05} is significantly lower than the best-known unweighted APSP
bound~\cite{Zwick02} of $\Ot(n^{2+\rho})$, where $\rho\approx 0.529$ is a number such that $\omega(1,\rho,1)=1+2\rho$ and
$\omega(a,b,c)$ is such that one can multiply $n^a\times n^b$ and $n^b\times n^c$ matrices in $O(n^{\omega(a,b,c)})$ time.
In fact, computing APSP in unweighted directed graphs is conjectured to require $\Theta(n^{2.5})$ time even if $\omega=2$~\cite{LincolnPW20} and there
are compelling reasons to believe that Zwick's algorithm is near-optimal~\cite{ChanWX21}.

\paragraph{Distance sensitivity oracles.}
Whereas the extreme no-preprocessing tradeoff transfers to fault-tolerant and dynamic settings
with no change, the precompute-all approach requires much more time and space effort
simply because there are much more possible queries to serve.
Indeed, a trivial solution would require precomputing $\Theta(n^2\cdot m^{f})$ distances
if at most $f$ edge failures are to be supported.
Despite this,~\cite{BernsteinK09} showed a 1-DSO for \emph{real-weighted} digraphs with $\Ot(nm)$ preprocessing time,
$\Ot(n^2)$ space
and constant query time. Note that their preprocessing matches the state-of-the-art $\Ot(nm)$ APSP
bound that is conjectured to be optimal for real-weighted digraphs.
\cite{DuanZ17a} improved the space bound of~\cite{BernsteinK09} to $O(n^2)$.
For the case $f=2$,~\cite{DuanP09a} gave a DSO with $\Ot(n^2)$ space
and $\Ot(1)$ query time, requiring higher polynomial preprocessing.

There is also extensive prior work on distance sensitivity oracles in the unweighted and small-integer-weights regimes.
\cite{WeimannY13} showed the first non-trivial distance sensitivity oracle in this setting.
For integer weights $\{-W,\ldots, W\}$, they could achieve subcubic preprocessing and subquadratic
query time for any constant number of failures.
In the same regime, \cite{GV20} showed the first 1-DSO with subcubic preprocessing and \emph{sublinear} query time.
\cite{ChechikC20} showed a 1-DSO with $O(Wn^{2.873})$ preprocessing and \emph{polylogarithmic} query time.
Around the same time,~\cite{Ren22} gave a 1-DSO for \emph{positive} edge weights with $O(Wn^{2.724})$ preprocessing
time and $O(1)$ query time.
The data structures~\cite{ChechikC20, GV20, Ren22, WeimannY13} are all randomized,
have a linear dependence on the largest
(absolute) edge weight, and also support path reporting.
\cite{BiloCC0S22} showed a derandomization of the approach of~\cite{Ren22} at the cost of slightly slower (but still subcubic)
preprocessing.

The aforementioned data structures for small weights all leverage fast matrix multiplication
to speed up combinatorial computations. DSOs with improved preprocessing and query times have been obtained
via a more aggressive use of algebraic techniques: forms of path counting~\cite{DemetrescuI05, KingS02} or
small-rank update to the matrix inverse~\cite{Sankowski04, Sankowski05, Sankowski05a}
combined with randomized polynomial identity testing~\cite{Zippel79}. These techniques
typically do not allow for efficient path reporting. Using algebraic techniques of this flavor,~\cite{GuR21} recently showed
a 1-DSO for digraphs with weights~$\{1,\ldots,W\}$ with $O(Wn^{2.58})$ preprocessing time that is quite close
to the $O(n^{2.529})$ APSP bound of~\cite{Zwick02}.

\cite{BrandS19} gave an algebraic DSO that can handle a \emph{polynomial} number of failures
in the case of weights $\{-W,\ldots,W\}$.
Specifically, for any $\mu\in [0,1]$, their data structure has $\Ot(Wn^{\omega+(3-\omega)\mu})$ construction time,
and after preprocessing a batch of $f$ failures in $\Ot(Wn^{2-\mu}f^2+Wnf^\omega)$ time,
answers distance queries wrt. that batch in $\Ot(Wn^{2-\mu}f+Wnf^2)$ time.
That is, if the failing edges are considered a part of the query, the query time
is $\Ot(Wn^{2-\mu}f^2+Wnf^\omega)$. In particular, one can handle up to $f=n^{1/\omega-\eps}\approx n^{0.42}$
failures with subcubic preprocessing and subquadratic query time.

\paragraph{Fully dynamic exact distance oracles.} For general weighted digraphs, \cite{DemetrescuI04}
gave a combinatorial deterministic fully dynamic data structure (later slightly improved by~\cite{Thorup04}) with $\Ot(n^2)$ amortized update
time maintaining all-pairs shortest paths explicitly. This
improves upon recompute-from-scratch for all but the sparsest digraphs.
The fully dynamic APSP problem (that is, explicitly maintaining the distance matrix)
has also been studied with the objective of optimizing the \emph{worst-case} update bounds~\cite{AbrahamCK17, ChechikZ23, GutenbergW20b, Thorup05}.
The current best-known worst-case update bound for APSP is $\Ot(n^{2+2/3})$ for weighted graphs
and $\Ot(n^{2.5})$ for unweighted graphs~\cite{AbrahamCK17, GutenbergW20b}.
In particular, the latter improves upon the static APSP bound of~\cite{Zwick02}.
Interestingly, all the known fully dynamic APSP data structures support vertex updates.

As far as fully dynamic exact distance oracles with a non-trivial query procedure are concerned,
\cite{RodittyZ11} showed a data structure tailored to sparse graphs
with $\Ot(m\sqrt{n})$ amortized (vertex) update time
and $\Ot(n^{3/4})$ query time, whereas~\cite{KarczmarzS23} recently presented
a data structure for real-weighted digraphs with $\Ot(mn^{4/5})$ worst-case update
time and $\Ot(n^{4/5})$ query time.

For dense graphs, dynamic distance oracles with both
subquadratic \emph{single-edge} update and query time can be obtained using variants of dynamic matrix inverse~\cite{Sankowski04, BrandNS19}.
The state-of-the-art worst-case update/query bound of this kind is~$\Ot(n^{1.703})$ due to~\cite{BrandFN22}.
Moreover, \cite{alokhina-brand, BergamaschiHGWW21} described shortest path-reporting extensions of these algebraic data structures with polynomially worse (but still subquadratic) worst-case update and query time.

\subsection{Our results}

\newcommand{\mm}{\operatorname{MM}}
\paragraph{Frobenius form toolbox.} We obtain two tools for generic
matrices (i.e., with a single invariant factor) over an arbitrary finite field $\field$ and accompanied by a Frobenius form. The first one is a data structure for querying
some number of initial powers of the matrix.

\begin{theorem}\label{t:frobenius-powers}
  Let $A\in \field^{n\times n}$ be a generic matrix and suppose its Frobenius normal form is given. One can preprocess $A$ in $\Ot(n^2)$ time so that the following queries are supported.

  Given $S,T\subseteq [n]$ and $h\in [1,n]$, compute the $S\times T$ submatrices
  of the matrix powers $A^1,\ldots,A^h$.
  The query time is $\Ot(n^{\omega(s,1-\alpha,t)+\alpha})$, where
  $|S|=\lfloor n^s\rfloor$, $|T|=\lfloor n^t\rfloor$ and $h=\lfloor n^\alpha\rfloor$.
\end{theorem}
Theorem~\ref{t:frobenius-powers} generalizes and improves upon a previous result of~\cite{SankowskiW19} who showed\footnote{This result does not depend on the erroneous statement in~\cite{SankowskiW19} about the graph diameter.} that after additional $\Ot(n^\omega)$-time preprocessing of a generic $A$ (given its Frobenius form), one can support queries $(i,j)$ asking for the $n$ values $(A^1)_{i,j},(A^2)_{i,j},\ldots,(A^n)_{i,j}$ in $\Ot(n)$ time.

One particularly important use case of Theorem~\ref{t:frobenius-powers} is computing the first $h\leq n^\alpha$ powers of~$A$. Theorem~\ref{t:frobenius-powers} implies that this is possible in $\Ot(hn^{\omega(1,1-\alpha,1)})$ time which polynomially improves upon
the trivial $O(hn^\omega)$ bound for all polynomial values of $h$.
To the best of our knowledge, previously, the first non-trivial result of this kind has been described for the case $h=n$: \cite{Storjohann15, ZhouLS15} showed that $n$ initial powers of~$A$ can be computed in $\Ot(n^3)$ time, which also follows from the data structure of~\cite{SankowskiW19}.
An improved $\Ot(h^2n^{\omega(1,1-\alpha,1-\alpha)})$ bound for computing the initial $h$ powers of~$A$ has been shown (implicitly) by~\cite{GuR21}. Both~\cite{Storjohann15,ZhouLS15}~and~\cite{GuR21} studied a more general problem of inverting an arbitrary degree-$d$ polynomial matrix modulo $x^{h+1}$.

We also show an improved dynamic algorithm updating a Frobenius form
of a generic matrix explicitly subject to a rank-1 perturbation.
\begin{restatable}{theorem}{trankone}\label{t:rank1}
  Let $A\in\field^{n\times n}$ be a generic matrix.
  Suppose an FNF of $A$ and an FNF of $A^T$ are given.
  Then, for any $a,b\in\field^{n\times 1}$ such that $A'=A+ab^T$ is generic,  Frobenius normal forms
  of $A'$ and $(A')^T$
  can be computed \emph{explicitly} in $\Ot(n^2)$ time. The algorithm succeeds with high probability. 
\end{restatable}
Here, the assumption that an FNF of the transpose is also given is without much loss of generality. Indeed, in a typical scenario, some FNF of $A$ is initialized in, say, $\Ot(n^\omega)$ time before the first application of Theorem~\ref{t:rank1}.
If we additionally compute an FNF of $A^T$ at that point within the same asymptotic bound, every subsequent application of Theorem~\ref{t:rank1} updates both FNFs.

The crucial advantage of Theorem~\ref{t:rank1} compared to the dynamic algorithm of~\cite{FrandsenS11} is that the FNFs -- including the (inverse) similarity transforms -- are updated explicitly.
This property is essential if we want to use Theorem~\ref{t:rank1} in combination with the data structure of Theorem~\ref{t:frobenius-powers} which requires an explicit FNF of the input matrix.

\paragraph{Applications to distance oracles.}
As an application of the developed tools for generic matrices with a Frobenius normal form, we show improved algebraic distance sensitivity oracles and fully dynamic
distance oracles for \emph{unweighted} dense directed graphs.

First of all, we show that for the 1-DSO problem, one can essentially match the static APSP bound
of~\cite{Zwick02} that is conjectured to be near-optimal~\cite{ChanWX21, LincolnPW20}.

\begin{restatable}{theorem}{ouronedso}\label{t:our-1-dso}
    Let $G$ be an unweighted digraph. In $\Ot(n^{2+\rho})=O(n^{2.529})$ time one can construct
    a distance sensitivity oracle for $G$ handling single-edge/vertex failures with $O(1)$ query time and $\Ot(n^2)$ space.
    The data structure is Monte Carlo randomized and the produced answers are correct with high probability\footnote{That is, with probability at least $1-1/n^c$, where
    the constant $c\geq 1$ can be set arbitrarily. We will also use the standard abbreviation w.h.p.}.
  \end{restatable}
\cite{GuR21} asked whether a 1-DSO with preprocessing time $\Ot(Wn^{2+\rho})$ is possible for graphs with weights $\{1,\ldots,W\}$. Theorem~\ref{t:our-1-dso} yields
an affirmative answer to this problem in the case $W=1$.

For distance oracles handling many failures, we show:

\begin{restatable}{theorem}{ourmanydso}\label{t:ourmanydso}
  Let $G$ be an unweighted digraph. There exists a distance sensitivity oracle with $\Ot(n^\omega)$ preprocessing and $O(n^2)$ space such that for any set $F$ of $f$ edge or vertex failures,
    the data structure can be updated in $\Ot(nf^{\omega-1})$ time
    to support distance queries with failures $F$ in $\Ot(nf)$ time.
    The data structure is Monte Carlo randomized and the produced answers are correct
    w.h.p.
  \end{restatable}
For unweighted digraphs, the data structure of Theorem~\ref{t:ourmanydso} improves upon the state-of-the-art~\cite{BrandS19} in terms of update
and query time even if~\cite{BrandS19} uses cubic preprocessing (i.e., if one sets $\mu=1$).
In particular, if the failures are part of the query, then we can handle
a distance query under up to $n^{1/(\omega-1)-\eps}\approx n^{0.72}$ failures
in subquadratic time.
Moreover, in the case of $f=\Ot(1)$ failures, our data structure has preprocessing and query
time matching the respective time characteristics of the state-of-the-art (failure-free) distance
oracle of~\cite{YusterZ05}.

That being said, the distance oracles of Theorems~\ref{t:our-1-dso}~and~\ref{t:ourmanydso} have some evident drawbacks compared to
the respective results of~\cite{GuR21, BrandS19}. Our data structure can be generalized to handle
small positive weights $[1, W]$ at the cost of introducing a multiplicative factor
polynomial in~$W$ -- by replacing $n$ with $nW$ in the preprocessing and query bounds.
That is, the respective preprocessing times in Theorems~\ref{t:our-1-dso}~and~\ref{t:ourmanydso} for weighted graphs
should be replaced with $\Ot((Wn)^{2.529})$ and $\Ot((Wn)^\omega)$, respectively.
On the other hand, the previously known data structures achieve a linear dependence on $W$.
Additionally, our data structure of Theorem~\ref{t:ourmanydso} does not seem to generalize
to negative edge weights, whereas that of~\cite{BrandS19} does.

Let us now move to our results in the dynamic scenario.
First, we obtain improved bounds for fully dynamic distance oracles supporting single-edge updates
in unweighted digraphs.

\begin{restatable}{theorem}{tweaked}\label{t:tweaked}
    Let $G$ be an unweighted digraph. There exists a Monte Carlo randomized data structure maintaining $G$ under
    single-edge insertions and deletions and supporting $s,t$-distance queries with
  $O(n^{1.673})$ worst-case update and query time. The answers produced are correct w.h.p.
  \end{restatable}
Theorem~\ref{t:tweaked} is obtained via a small tweak to the data structure of~\cite{BrandFN22} using Theorem~\ref{t:frobenius-powers}. Interestingly, if
$\omega=2$, the update/query bounds of both data structures (Theorem~\ref{t:tweaked} and that of~\cite{BrandFN22})
simplify to an odd-looking bound of $\Ot(n^{1+5/8})=\Ot(n^{1.625})$.

One component of~\cite{BrandFN22} is periodically recomputing bounded-hop all-pairs distances.
Using a different approach avoiding this entirely, we obtain another dynamic distance oracle.
\begin{restatable}{theorem}{dyndist}\label{t:dyndist}
  Let $G$ be an unweighted digraph. There exists a Monte Carlo randomized data structure maintaining $G$ under
  single-edge insertions and deletions and supporting $s,t$-distance queries with
  $\Ot\left(n^{\frac{\omega+1}{2}}\right)$ worst-case update and query time. The answers produced are correct w.h.p.
\end{restatable}
The update/query bound of $\Ot\left(n^{(\omega+1)/2}\right)=O(n^{1.687})$ is currently inferior to that of Theorem~\ref{t:tweaked} but nevertheless superior
to the state-of-the-art bound $O(n^{1.703})$~\cite{BrandFN22}. However, the data structure of Theorem~\ref{t:dyndist}
might be considered more promising: if $\omega=2$, its update bound simplifies to a natural $\Ot(n^{1.5})$ bound, and
even if $\omega<2.25$, the $\Ot(n^{(\omega+1)/2})$ bound is better than the theoretical limit of the approach of~\cite{BrandFN22}.

Finally, we achieve a very natural tradeoff in the more general case of vertex updates.
\begin{restatable}{theorem}{vertexupd}\label{t:vertexupd}
  Let $G$ be an unweighted digraph. There exists a Monte Carlo randomized data structure maintaining $G$ under fully
  dynamic vertex updates in $\Ot(n^2)$ worst-case time per update and supporting
  arbitrary pair distance queries in $\Ot(n)$ time. The answers are correct w.h.p.
\end{restatable}
Theorem~\ref{t:vertexupd} shows that one can preserve the linear query of the static distance oracle of~\cite{YusterZ05}
without rebuilding it from scratch. Similarly as in the case of the previously known data structures
supporting vertex updates for transitive closure~\cite{Sankowski04} and APSP, our data structure
does not need fast matrix multiplication for performing updates or queries within the stated bounds.
Another interesting consequence of Theorem~\ref{t:vertexupd} is that we can maintain distances between
$n$ \emph{arbitrary} pairs of vertices in $\Ot(n^2)$ \emph{worst-case} time per update. To the best of our knowledge,
no previous data structure could achieve that for unweighted dense graphs.
\cite{Karczmarz21}~showed that $\Ot(mn^{2/3})$ worst-case update time is possible for sparse weighted digraphs
if the pairs of interest are fixed.

Similarly as for distance oracles, our dynamic data structures generalize to digraphs with
weights $\{1,\ldots,W\}$ at the cost of a $o(W^2)$ factor in the respective bounds.
Again, the previous best bound~\cite{BrandFN22} can be easily lifted to this case
with an overhead linear in $W$.

\newcommand{\xvar}{\tilde{x}}
\newcommand{\yvar}{\tilde{y}}
\newcommand{\xvarset}{\tilde{X}}
\newcommand{\yvarset}{\tilde{Y}}
\newcommand{\asym}{\tilde{A}}

\subsection{Technical overview}

\paragraph{Generic matrices and Frobenius form.}
Let $U\cdot C\cdot U^{-1}$ be a Frobenius form of a generic matrix~$A\in\field^{n\times n}$,
where $C\in\field^{n\times n}$ is the companion matrix of the characteristic polynomial of $A$,
and $U\in\field^{n\times n}$ is a (not necessarily unique) invertible similarity transform.
See Section~\ref{s:frobenius} for more detailed definitions of these notions.
As observed
by~\cite{SankowskiW19}, since $C$ is a companion matrix:
\begin{enumerate}[(1)]
  \item the $n$ matrices $U\cdot C,U\cdot C^2,\ldots, U\cdot C^n$
    can be computed in $\Ot(n^\omega)$ time and stored explicitly (albeit succinctly) in $O(n^2)$ space;
\item
storing these matrices enables computing the values $A_{i,j},(A^2)_{i,j},\ldots,(A^n)_{i,j}$
    for any query pair $(i,j)$ in $\Ot(n)$ time via Hankel matrix-vector
multiplication which in turn is easily reducible to polynomial
    multiplication, i.e., FFT~\cite{fft}. See, e.g.,~\cite{golub2013matrix}.
\end{enumerate}
To obtain Theorem~\ref{t:frobenius-powers}, we first show a more efficient $\Ot(n^2)$-time algorithm for computing the matrices $U\cdot C,\ldots,U\cdot C^n$ by interpreting this problem as generation of multiple terms
of linear recurrences of order $n$ and employing an efficient recent algorithm for this task~\cite{BostanM21}.

Next, we generalize item~(2) above to queries about an arbitrary submatrix of some $h\leq n$ initial powers. While for small values of $h$, say $h=\sqrt{n}$, we cannot evaluate the cells $(i,j)$
of $h$ initial powers faster than in $\Ot(n)$ time,
we observe that considerable computational savings (on average) are possible if a larger submatrix $S\times T$ (potentially the full $n\times n$ submatrix), for $S, T\subseteq [n]$, is queried.
Indeed, in such a case, a careful packing of the matrices $U\cdot C, U\cdot C^2,\ldots, U\cdot C^n$
into two degree-$\Theta(h)$ polynomial matrices of sizes $|S|\times \lceil n/h\rceil$ and $\lceil n/h\rceil\times |T|$ (resp.)  allows us
to benefit from FFT~\cite{fft} and fast rectangular matrix multiplication~\cite{GallU18, HuangP98} at the same time.

\paragraph{Dynamic Frobenius form.}
To obtain the improved dynamic Frobenius form data structure of Theorem~\ref{t:rank1}, more ideas are needed.
Crucially, we use a notion of a \emph{generic vector} wrt. $A$, i.e., a vector $u\in \field^{n\times 1}$ such that the \emph{iterates} $u,Au,\ldots,A^{n-1}u$ are linearly independent. In other words, $u$ is generic wrt. $A$ if the order-$n$ Krylov subspace generated by $A$ and $u$ has dimension $n$.
As shown by~\cite{Keller-Gehrig85}, if $u$ is generic wrt. $A$, the iterates
of $u$ encode a similarity transform~$U$ such that $U\cdot C\cdot U^{-1}$
is a Frobenius form of $A$ (and $C$ is a companion matrix of $p_A$). Using the techniques of~\cite{BrentGY80, eberly00}, one can in fact prove that, given the iterates of a generic vector $u$ wrt. $A$
and iterates of a generic vector $v$ wrt. $A^T$,
all three matrices $U,C,U^{-1}$ comprising the Frobenius form can
be computed in $\Ot(n^2)$ time.
Moreover, that algorithm can be used to detect non-genericity of input vectors: if either $u$ or~$v$ is not generic wrt. the respective matrix, the algorithm fails.
Therefore, if $A$ is subject to a rank-$1$ update
$A':=A+ab^T$ (which keeps $A'$ generic), then in order to compute a Frobenius
form of $A'$ explicitly in $\Ot(n^2)$ time, it is sufficient to compute
the iterates of some generic vectors $u',v'$ wrt. the matrices $A+ab^T$ and $(A+ab^T)^T$ in $\Ot(n^2)$ time.

We show a dynamic programming-based algorithm for this task that works even if $u',v'$ are not generic. Specifically, given FNFs of generic $A$ and $A^T$ and some arbitrary vectors $u',v'\in\field^{n\times 1}$, the algorithm
computes the iterates of $u'$ wrt. $A+ab^T$ and $v'$ wrt. $(A+ab^T)^T$.
One important ingredient here is the preprocessing of Theorem~\ref{t:frobenius-powers} which also allows computing iterates of an arbitrary vector in near-optimal $\Ot(n^2)$ time.
With the respective iterates wrt. $A$ and $A^T$ in hand, the obtained dynamic programming formula for the subsequent iterates wrt. the perturbed matrices
can be efficiently evaluated using a folklore combination of divide-and-conquer and FFT.

Finally, as proved by~\cite{BrentGL03}, generic vectors wrt. a generic $n\times n$ matrix over a finite field can be obtained (w.h.p.) within $\Ot(\polylog{n})$ random samples even for small fields.
Thus, computing the iterates of $O(\polylog{n})$ random vectors wrt. $A+ab^T$ and $(A+ab^T)^T$ and feeding them into the aforementioned procedure based on~\cite{BrentGY80, eberly00} yields an $\Ot(n^2)$-time Las Vegas randomized (w.h.p.) algorithm computing an FNF of a (generic) matrix $A$ after a rank-1 perturbation.

\paragraph{From generic matrices to graphs.} The key technical idea enabling all our developments for fault-tolerant and dynamic distance oracles
is to represent an \emph{arbitrary} directed graph using a \emph{generic matrix}.

Roughly speaking, the state-of-the-art dynamic distance oracle~\cite{BrandFN22} for unweighted graphs (and its predecessors~\cite{BrandNS19, Sankowski05})
rely on \emph{path counting} of sufficiently short paths.
If $A(G)$ is an adjacency graph of $G$, then $(A(G)^k)_{s,t}$
equals the number of distinct $s\to t$ paths consisting of precisely~$k$ edges.
As a result, if $(A(G)^k)_{s,t}\neq 0$,
the distance from $s$ to $t$ is no more than~$k$.
In the other direction, if an $s\to t$ path of length $k$ exists in $G$,
then ${(A(G)^k)_{s,t}\neq 0}$.
Consequently, the matrix powers $A(G),A(G)^2,\ldots,A(G)^h$ encode the \emph{short} distances
between vertices at distance at most $h$.
The short distances (for a sublinear $h=\poly(n)$), combined with standard hitting set arguments~\cite{UY91},
already allow computing an $s,t$ distance in $G$ in subquadratic time.

The challenge is to efficiently compute the first $h$ matrix powers of $A(G)$ 
and maintain them (possibly implicitly) under element updates to $A(G)$.
The simple-minded approach leads to an $\Ot(n^\omega\cdot h^2)$ time for the static computation
since the elements of $A(G)^h$ may use up to $\Ot(h)$ bits.
In the path applications, we are only interested
in whether the entries of the powers are zero or not, so
performing all the counting modulo a sufficiently large
random $\Ot(\polylog{n})$-bit prime still yields high-probability correctness.
Therefore, the powers $A(G),\ldots,A(G)^h$ can be thought to be computable
in $\Ot(n^\omega\cdot h)$ time.
The well-established way to handle updates is to note that these powers are encoded
by the inverse of the polynomial matrix $I-A(G)\cdot X$ in the ring of polynomials
modulo $X^{h+1}$ and apply dynamic matrix inverse data structures~\cite{Sankowski04, BrandNS19, BrandFN22}.
Since the entries of $A(G)^h$ are polynomials of degree at most $h$, the obtained update
bounds are generally factor-$\Ot(h)$ away from the known dynamic matrix inverse bounds~\cite{Sankowski04, BrandNS19}.

As shown by~\cite{Storjohann15, ZhouLS15}, the $\Ot(n^\omega\cdot h)$ bound for computing the first $h$ powers
statically is certainly \emph{not optimal} for large values of $h$:
for $h=n$ the computation can be performed in near-optimal $\Ot(n^3)$ time
precisely via a reduction to polynomial matrix inverse.
In fact, the state of the art DSOs~\cite{BrandS19, GuR21} that we improve upon
rely on the techniques of~\cite{ZhouLS15}.
Specifically, one of the contributions of~\cite{GuR21} is showing that $h=n^\alpha$ first powers
of a matrix can be computed in $\Ot(n^{\omega(1,1-\alpha,1-\alpha)+2\alpha})$ field operations.

To obtain the improved distance oracles, we avoid using the techniques of~\cite{ZhouLS15} and apply our data structure for querying matrix powers (Theorem~\ref{t:frobenius-powers}) to compute the $h$-bounded distances faster, in $\Ot\left(hn^{\omega(1,1-\alpha,1)}\right)$ time (after $\Ot(n^\omega)$-time computation of an FNF).
Such a speed-up alone is enough to obtain
the tweaked dynamic distance oracle of~Theorem~\ref{t:tweaked},
and, via the reduction of~\cite{GuR21}, a 1-DSO with preprocessing time matching
the APSP bound of~\cite{Zwick02}.

For the above application of Theorem~\ref{t:frobenius-powers} to be legitimate, we need to guarantee that the graph is represented
by a generic matrix.
Simply using a standard adjacency matrix of a digraph fails here,
since the adjacency matrix is often not generic, i.e., it may have multiple invariant factors.
Based on the techniques of~\cite{Wiedemann86}, we show that a \emph{weighted adjacency matrix} $A(G)\in\field^{n\times n}$ representing~$G$ can be appropriately and efficiently sampled so that $A(G)$
is generic, if the size of the field $\field$ used is sufficiently large but still polynomial in $n$.
A random weighted adjacency matrix also has, with high probability, the desired properties relating the
non-zero entries of $(A(G))^k$ to the existence of length-$\leq k$ paths between vertex pairs. Moreover, the elements of $A(G)$ are sampled independently and thus
small updates to the graph~$G$ result in small updates to $A(G)$ and they maintain
genericity with high probability (over polynomially many updates).

A further combination of the submatrix queries data structure of Theorem~\ref{t:frobenius-powers} with formulas for updating the matrix inverse after changing few elements~\cite{BrandNS19, Sankowski04, Sankowski05} and standard hitting set
arguments~\cite{UY91} allows us to obtain the improved multiple-failures distance sensitivity oracle of Theorem~\ref{t:ourmanydso}
and the ``prospective'' dynamic distance oracle of Theorem~\ref{t:dyndist}.

We note that the previous work, in particular concerning
(possibly negatively-) weighted graphs (e.g.,~\cite{BrandS19, Sankowski05a})
or reachability in presence of cycles~\cite{Sankowski04}, where simple path counting fails,
also used weighted adjacency matrices, albeit for a different reason. There, one starts with a polynomial \emph{symbolic adjacency matrix}
in the first place and then applies random variable substitution to enable
efficient polynomial identity testing~\cite{Schwartz80,Zippel79}.
Our use of weighted adjacency matrices can be considered
merely a trick to fix the non-genericity in basic path counting.

\subsection{Further related work}
Specialized exact distance oracles have been shown for incremental~\cite{AusielloIMN91} and decremental~\cite{BaswanaHS07, EvaldFGW21} unweighted directed graphs.
Non-trivial \emph{approximate} distance oracles for weighted directed graphs are known in the fully dynamic setting~\cite{BrandN19}
and partially dynamic settings~\cite{Bernstein16, EvaldFGW21, KarczmarzL19}.

There has been extensive and influential work on static distance oracles for \emph{undirected graphs}, especially in various approximate settings,
e.g.,~\cite{Chechik15, ChechikZ22, PatrascuR14, ThorupZ05, Wulff-Nilsen12}. See also the survey~\cite{Sommer14}. Distance oracles for undirected graphs
have also been studied specifically in the
fault-tolerant (e.g.,~\cite{ChechikLPR12, DuanR22}) and fully dynamic (e.g.,~\cite{Bernstein09, BrandN19}) settings.

\section{Preliminaries}
We denote by $[n]$ the set $\{1,\ldots,n\}$.
Let $\field$ be a finite field.
For an $n\times m$ matrix~$A\in \field^{n\times m}$, and $S\subseteq [n],T\subseteq [m]$, we generally denote by $A_{S,T}$ the submatrix
of $A$ with rows $S$ and columns $T$. We may write $A_{s,T}$ or $A_{S,t}$, for $s,t\in [n]$, to denote $A_{\{s\},T}$ or $A_{S,\{t\}}$, respectively.
In particular, $A_{s,t}$ is the element in the cell $(s,t)$ of $A$.
Whenever we write $A_S$, we mean $A_{S,S}$.

If $v$ is a column (row) vector in $\field^{n\times 1}$ (in $\field^{1\times n}$, resp.), then we sometimes write $v_i$ to denote $v_{i,1}$ ($v_{1,i}$, resp.).
If $n$ is known from the context, we denote by $e_i\in \field^{n\times 1}$ a column vector satisfying $(e_i)_j=[j=1]$.

We generally measure time in field operations, i.e., the field operations are assumed to take unit time.
We denote by $\mm(p,q,m)$ the time needed to multiply a matrix from $\field^{p\times q}$ by a matrix from $\field^{q\times m}$.
That is, if $p=\lfloor n^\alpha\rfloor$, $q=\lfloor n^\beta \rfloor$, $m=\lfloor n^\gamma \rfloor$ for $\alpha,\beta,\gamma\geq 0$,
then $\mm(p,q,m)=O(n^{\omega(\alpha,\beta,\gamma)})$.

When talking about directed graphs $G=(V, E)$, for $F\subseteq E\sqcup V$ we denote by $G-F$ the graph obtained from $G$ by removing the vertices and/or edges $F$.
If $G$ is (non-negatively) weighted, then we denote by $\wei_G(uv)$ the weight of the edge $uv\in E$.
For any $s,t\in V$, denote by $\dist_G(s,t)$ the weight of the shortest $s\to t$ path in~$G$.
If no $s\to t$ path exists in $G$, we put $\dist_G(s,t)=\infty$.

\section{Generic matrices and Frobenius form}\label{s:frobenius}
Let $\field$ be a finite field.
Once again, we call a matrix $A\in \field^{n\times n}$ \emph{generic}
if the \emph{characteristic polynomial} \linebreak $p_A(t)=\det(tI-A)$
equals the \emph{minimal polynomial} $\mu_A$ of~$A$, i.e.,
the minimum-degree monic polynomial over $\field$ such that $\mu_A(A)\equiv 0$.
We start with the following well-known fact.
\begin{fact}{\upshape (see, e.g.,~\cite{basicalgebra})}\label{f:frobenius}
  Suppose the matrix $A$ is generic and let $p_A(t)=t^n+c_{n-1}t^{n-1}+\ldots+c_0$.
  There exists an invertible matrix $U\in \field^{n\times n}$ such that
  \begin{equation}\label{eq:frob}
    A=U\cdot C\cdot U^{-1},
  \end{equation}
  where $C\in \field^{n\times n}$ is a \emph{companion matrix} of $p_A$, that is:
  \begin{equation*}
    C=\begin{bmatrix}
      0 & 0 & \dots & 0 & -c_0\\
      1 & 0 & \dots & 0 & -c_1\\
      0 & 1 & \dots & 0 & -c_2\\
      \vdots & \vdots & \ddots & \vdots\\
      0 & 0 & \dots & 1 & -c_{n-1}
    \end{bmatrix}.
  \end{equation*}
\end{fact}
For any \emph{similarity transform} $U\in\field^{n\times n}$ satisfying Equation~\eqref{eq:frob},
$U\cdot C\cdot U^{-1}$ is called the \emph{Frobenius normal form} (FNF) of the generic matrix $A$.
The Frobenius form can be defined more generally for arbitrary non-generic matrices from $\field^{n\times n}$, albeit the middle matrix $C$ has then a more complicated form -- it may consist of multiple companion matrices.
FNF can be computed deterministically in $\Ot(n^\omega)$ \cite{Storjohann01} time for every matrix from $\field^{n\times n}$.
For generic matrices, there exists an easier $\Ot(n^\omega)$-time FNF algorithm~\cite{Keller-Gehrig85} that we sketch
in this section and build upon later on.

Let us call a vector $u\in \field^{n\times 1}$ \emph{generic wrt. $A$} if
the vectors $u,Au,A^2u,\ldots,A^{n-1}u$ are linearly independent.

\begin{lemma}\upshape{\cite{Keller-Gehrig85}}\label{l:keller}
  Let $A\in \field^{n\times n}$ be a generic matrix and let $u\in \field^{n\times 1}$ be generic wrt. $A$.
  Then the matrix $U=\begin{bmatrix}u&Au&A^2u&\ldots&A^{n-1}u\end{bmatrix}\in \field^{n\times n}$ is invertible
  and $U\cdot C\cdot U^{-1}$ is an FNF of $A$.
\end{lemma}
\begin{proof}
  $U$ is invertible since it is of size $n\times n$ and its columns are linearly independent.
  Moreover, we have $AU=[Au|A^2u|\ldots|A^nu]$. On the other hand, for any vector $v\in\field^{1\times n}$ such that
  \linebreak
  $v=(v_0,\ldots,v_{n-1})$, we have
  $vC=(v_1,\ldots,v_{n-1},-\sum_{i=0}^{n-1}c_iv_i)$.
  As a result:
  \begin{equation*}
    UC=\begin{bmatrix}Au&A^2u&\ldots&A^{n-1}u&-\left(\sum_{i=0}^{n-1}c_iA^i\right)u\end{bmatrix}.
  \end{equation*}
  Since $A$ is generic, $\mu_A(A)=p_A(A)=A^n+\sum_{i=1}^{n-1}c_iA^i\in \field^{n\times n}$ is a zero matrix.
  Consequently, we have $UC=\begin{bmatrix}Au&A^2u&\ldots&A^nu\end{bmatrix}$, which proves $AU=UC$. Thus, indeed $A=UCU^{-1}$.
  \end{proof}

  We now refer to~\cite{BrentGL03} for the following estimate.
\begin{theorem}\label{t:random-vector}\upshape{\cite[Theorem~9]{BrentGL03}}
  Let $A\in \field^{n\times n}$ be a generic matrix and let $q=|\field|$. Then, with probability at least $\frac{0.2}{1+\log_q{n}}$, a random
  vector from $\field^{n\times 1}$ is generic wrt. $A$.
\end{theorem}
Since one can compute the $n$ vectors $u,Au,\ldots,A^{n-1}u$ and matrix inverse
in general in $\Ot(n^\omega)$ time (see, e.g.,~\cite{Keller-Gehrig85}),
by Theorem~\ref{t:random-vector},
Lemma~\ref{l:keller} applied to random vectors $u\in \field^{n\times 1}$ yields:
\begin{lemma}\label{l:fnf-simple}
  Let $A\in \field^{n\times n}$ be a generic matrix. There is a Las Vegas algorithm computing an FNF of $A$ in $\Ot(n^\omega)$ time.
  The running time bound holds with high probability.
\end{lemma}
\begin{proof}
  Multiplying a random vector $u\in\field^{n\times 1}$ by the $n$ first powers of $A$ yields an FNF of $A$
  with probability $\Omega(1/\log{n})$. After $c\cdot \log^2{n}=O(\polylog{n})$ trials (where $c=O(1)$), the success probability is at least $1-O(n^{-c})$.
\end{proof}
Crucially for our applications,
an FNF can be computed faster if an efficient way of multiplying a vector by $n$ powers of $A$ is available.
The following lemma has been proven in a more general form by Eberly~\cite{eberly00}. We include a proof for completeness.
\begin{lemma}\label{l:eberly}{\upshape{\cite{eberly00}}}
  Let $A\in \field^{n\times n}$ be a generic matrix, and let $u,v\in \field^{n\times 1}$.
  Suppose the vectors $u,Au,\ldots,A^{n-1}u$ and $v^T,v^TA,\ldots,v^TA^{n-1}$ are given.

  Then, using $\Ot(n^2)$
  additional field operations one can either compute a Frobenius normal form
  of $A$ or detect that either $u$ is not generic wrt. $A$ or $v$ is not generic wrt. $A^T$.
\end{lemma}
\begin{proof}
  We wish to compute the Frobenius form $U\cdot C\cdot U^{-1}$ given by Lemma~\ref{l:keller} or detect
  that either $u$ is not generic wrt. $A$ or $v$ is not generic wrt. $A^T$. Recall that the matrix
  $U$ is obtained by putting the (given) vectors $u,Au,\ldots,A^{n-1}u$ in a row.
  Moreover, let
  \begin{equation*}
    V=\begin{bmatrix}v^T\\v^TA\\\vdots\\v^TA^{n-1}\end{bmatrix}.
  \end{equation*}
  Consider a Hankel matrix (i.e., with all skew diagonals constant)
  \begin{equation*}
    VU=\begin{bmatrix}
      v^Tu & v^TAu & \ldots & v^TA^{n-1}u\\
      v^TAu & v^TA^2u & \ldots & v^TA^nu\\
      \vdots & \vdots & \ddots & \vdots\\
      v^TA^{n-1}u & v^TA^nu & \ldots & v^TA^{2n-2}u
    \end{bmatrix}.
  \end{equation*}
  Note that the $2n-1$ distinct entries of $VU$ can be computed
  using $O(n)$ inner products of the input vectors, i.e., in $O(n^2)$ time.
  One can check in $\Ot(n)$ time whether an $n\times n$ Hankel matrix (given its $O(n)$ skew diagonal values)
  is singular~\cite{BrentGY80}.
  If $VU$ is singular, then either $V$ or $U$ is singular, that is, either $v$ is not generic
  wrt. $A^T$ or $u$ is not generic wrt. $A$.

  Suppose both $U$ and $V$ are non-singular. We have
  \begin{equation*}
    U^{-1}=(VU)^{-1}\cdot V,
  \end{equation*}
  so the columns of $U^{-1}$ can be found by solving $n$ linear systems $Hx=b$, where $H=VU$ is
  a non-singular Hankel matrix and $b$ is a column of $V$. Each such linear system can be solved
  using $\Ot(n)$ field operations~\cite{BrentGY80}. As a result, $U^{-1}$ can be computed in $\Ot(n^2)$ time.

  Finally, since $C=U^{-1}\cdot A\cdot U$ and all but the last column of $C$ are fixed,
  we can determine~$C$ by simply computing the last column of $U^{-1}\cdot A\cdot U$,
  i.e., multiplying $U^{-1}\cdot A\cdot A^{n-1}u$ in $O(n^2)$ time.
\end{proof}
Lemma~\ref{l:eberly} can produce an FNF of a generic
$A$ in $\Ot(T(n)+n^2)$ time, where $T(n)$ is the time required to multiply
a vector by the first $n$ powers of either $A$ or $A^T$.
Note that $A^T$ and $A$ have the same characteristic and minimal polynomials,
so $A^T$ is generic iff $A$ is generic.
Thus, by Theorem~\ref{t:random-vector}, $O(\polylog{n})$ samples of $u,v$ are enough to succeed w.h.p.

The crucial property of the Frobenius normal form that we will use is that the companion
matrix $C$ is particularly easy to power
and for any $k\geq 1$ we have
\begin{equation*}
  A^k=(UCU^{-1})^k=U(CU^{-1}U)^{k-1}CU^{-1}=UC^kU^{-1}.
\end{equation*}

The companion matrix $C$ has the following key property.
\begin{fact}\label{f:companion-gen}
  For any $k\geq 1$, let $w_1,w_2,\ldots,w_n\in\field^{n\times 1}$
  be the columns of $C^k$. Then we have
  \begin{equation*}
    C^{k+1}=\begin{bmatrix}
      w_2 & w_3 & \dots & w_n & C\cdot w_n
    \end{bmatrix}.
  \end{equation*}
\end{fact}

\begin{corollary}\label{cor:upow-gen}
  Let $k\geq 1$
  and let $u_1,\ldots,u_n$ be the columns of the matrix $U\cdot C^k$. Let $w_n$ be the last column of $C^k$. Then
  \begin{equation*}
    U\cdot C^{k+1}=\begin{bmatrix}
      u_2 & u_3 & \dots & u_n & U\cdot C\cdot w_n
    \end{bmatrix}.
  \end{equation*}
\end{corollary}
\begin{proof}
Let $w_1,w_2,\ldots,w_n\in\field^{n\times 1}$ be the columns of $C^k$. Then, by Fact~\ref{f:companion-gen}:
\begin{equation*}
    U\cdot C^{k+1}=\begin{bmatrix}
      Uw_2 & Uw_3 & \dots & Uw_n & UCw_n
    \end{bmatrix}=
    \begin{bmatrix}
      u_2 & u_3 & \dots & u_n & UC w_n
    \end{bmatrix}.\qedhere
  \end{equation*}
\end{proof}
By Corollary~\ref{cor:upow-gen}, the $n$ matrices $U\cdot C$, $U\cdot C^2$, \ldots, $U\cdot C^n$ can
be encoded concisely using only $2n-1$ column vectors, i.e., in $O(n^2)$ space.
This is formally captured by the below lemma which also
shows that such a representation can be computed very efficiently.
\begin{lemma}\label{l:rep}
  Let the Frobenius normal form $U\cdot C\cdot U^{-1}$ of $A$ be given.
  Using
  $\Ot(n^2)$
  field operations we can compute an \emph{auxiliary matrix}
  $(r_{i,j})=R\in \field^{n\times (2n-1)}$, such that for any $k\in[n]$:
  \begin{equation*}
    U\cdot C^k=\begin{bmatrix}
      r_{1,k} & r_{1,k+1} & \dots & r_{1,n+k-1}\\
      r_{2,k} & r_{2,k+1} & \dots & r_{2,n+k-1}\\
      \vdots & \vdots & \ddots & \vdots\\
      r_{n,k} & r_{n,k+1} & \dots & r_{n,n+k-1}
    \end{bmatrix}:=R_k.
  \end{equation*}
\end{lemma}
\begin{proof}
  Let $u_i=(u_{i,1},\ldots, u_{i,n})$ be the $i$-th row of the matrix $U$.
  For any $k>n$ define $u_{i,k}$ inductively:
  \begin{equation*}
    u_{i,k}=-c_0\cdot u_{i,k-n}-c_1\cdot u_{i,k-n+1}-\ldots-c_{n-1}\cdot u_{i,k-1}=-\sum_{i=0}^{n-1}c_i\cdot u_{i,k-n+i}.
  \end{equation*}
  In other words, $(u_i)_{i=1}^\infty$ is a linearly recursive sequence of order $n$.
  As discussed in the proof of Lemma~\ref{l:keller}, for any $k\geq 0$, we have:
  \begin{equation*}
    u_i\cdot C^k=(u_{i,1+k},\ldots,u_{i,n+k}).
  \end{equation*}
  In particular, for $k=n-1$, we obtain:
  \begin{equation*}
    u_i\cdot C^{n}=(u_{i,n+1},\ldots,u_{i,2n})=(r_{i,n},\ldots,r_{i,2n-1}).
  \end{equation*}
  Since $(r_{i,1},\ldots,r_{i,n-1})=(u_{i,2},\ldots,u_{n})$, the $i$-th row of the
  matrix $R$ can be obtained by computing the terms $n+1,\ldots,2n$ of
  the linearly recursive sequence $(u_i)_{i=1}^\infty$.
  As shown in~\cite[Theorem~3]{BostanM21}, this can be done using $\Ot(n)$ field operations.
  Therefore, by applying this to all $i$, computing the entire matrix $R$
  is possible using $\Ot(n^2)$ field operations.
\end{proof}

\begin{lemma}\label{l:powers-vector}
  Let a Frobenius normal form $U\cdot C\cdot U^{-1}$ of a generic $A\in\field^{n\times n}$ and the associated auxiliary matrix $R$ of Lemma~\ref{l:rep} be given. Then:
  \begin{enumerate}[label=(\arabic*)]
    \item For any $i,j\in [n]$, the elements $A_{i,j},(A^2)_{i,j},\ldots,(A^{n-1})_{i,j}$ can be computed in $\Ot(n)$~time.
    \item For any vector $v\in \field^{n\times 1}$, all the vectors $v,Av,A^2v,\ldots,A^{n-1}v$ can be computed in $\Ot(n^2)$~time.
  \end{enumerate}
\end{lemma}
\begin{proof}
  In the former item, set $v:=e_j$.
  We have $A^kv=U\cdot C^k\cdot (U^{-1}v)$.
  Let us first compute $w=U^{-1}v$.
  In the former item, the vector $w$ is simply the $j$-th column of $U^{-1}$ and thus it can be
  read in $O(n)$ time.
  In the latter item, it can be obtained in $O(n^2)$ time.

  For any $i\in [n]$, $(A^kv)_i=(UC^kw)_i=(r_{i,k},\ldots r_{i,n+k-1})\cdot w$.
  Thus, we have:
  \begin{equation*}
    \begin{bmatrix}
      r_{i,1} & r_{i,2} & \dots & r_{i,n}\\
      r_{i,2} & r_{i,3} & \dots & r_{i,n+1}\\
      \vdots & \vdots & \ddots & \vdots\\
      r_{i,n} & r_{i,n+1} & \dots & r_{i,2n-1}
    \end{bmatrix}\cdot w=\begin{bmatrix}
      (A^1v)_i\\
      (A^2v)_i\\
      \vdots\\
      (A^nv)_i
    \end{bmatrix}.
  \end{equation*}
  Note that the $n\times n$ matrix on the left-hand side above is a Hankel matrix whose
  $2n-1$ distinct entries come from the precomputed matrix $R$.
  As a result, the right-hand side vector can be computed using fast Hankel matrix-vector
  multiplication in $\Ot(n)$ time~(see, e.g.,~\cite{golub2013matrix}).
  This gives the desired values $A_{i,j},(A^2)_{i,j},\ldots,(A^{n-1})_{i,j}$ in item~(1).
  By doing this for all $i=1,\ldots,n$, we obtain the desired
  vectors $v,Av,\ldots,A^{n-1}v$ in $\Ot(n^2)$ time in item~(2).
\end{proof}

\section{Computing submatrices of $k$ first powers of a generic matrix}
Let $A\in \field^{n\times n}$ be a generic matrix.
As shown in Lemma~\ref{l:powers-vector}, one can compute a certain cell
$(i,j)$ of all the powers $A^1,\ldots,A^{n-1}$ in $\Ot(n)$ time via Hankel matrix-vector
multiplication as long as a Frobenius form of $A$
is given.
In this section, we generalize this as follows.
Let $1\leq h\leq n$ be an integer.
Let $S$ be a subset of rows and let~$T$ be a subset of columns of~$A$.
Our goal is to compute the $S\times T$ submatrices $(A^1)_{S,T},\ldots,(A^h)_{S,T}$.
One particularly interesting case is $S=T=\{1,\ldots,n\}$, where we want to explicitly
output the first $h$ matrix powers of $A$.
We prove:

\begin{theorem}\label{t:powers}
  Let $A\in \field^{n\times n}$ be generic and let $U\cdot C\cdot U^{-1}$ be its Frobenius form.
  Let $R=(r_{i,j})$ be the associated auxiliary matrix of Lemma~\ref{l:rep}.
  Let $S,T\subseteq [n]$.
  Let $h=\lfloor n^{\alpha}\rfloor$, $|S|=\lfloor n^{s}\rfloor$, ${|T|=\lfloor n^t\rfloor}$ for some $\alpha,s,t\in [0,1]$.
  Then, the submatrices $(A^1)_{S,T},\ldots,(A^h)_{S,T}$ can be computed
  using $\Ot(n^{\omega(s,1-\alpha,t)+\alpha})$ field operations.
\end{theorem}
\begin{proof}
  Set $\Delta=\lceil n/h\rceil$.
  Put $U^{-1}=(g_{i,j})$. For all $i\in S$ and $j\in \{0,\ldots,\Delta-1\}$, let
  \begin{equation*}
    p_{i,j}(x)=r_{i,j\cdot h+1}\cdot x+r_{i,j\cdot h+2}\cdot x^2+\ldots +r_{i,j\cdot h+(2h-1)}\cdot x^{2h-1}
  \end{equation*}
  be a polynomial. Similarly, for all $i\in T$ and $j\in\{0,\ldots,\Delta-1\}$, let us introduce a polynomial
  \begin{equation*}
    q_{j,i}(x)=g_{j\cdot h+1,i}\cdot x^{h-1}+g_{j\cdot h+2,i}\cdot x^{h-2}+\ldots + g_{j\cdot h+h,i}\cdot x^0.
  \end{equation*}
  In the above, every value $r_{\cdot,\cdot}$ and $g_{\cdot,\cdot}$ that has not been defined
  is assumed to be equal to $0$.

  Each $p_{i,j}$ is a polynomial of degree $2h-1$, and each $q_{j,i}$ is a polynomial of degree $h-1$.
  Consider the polynomial matrices $P=(p_{i,j})\in \field[x]^{|S|\times \Delta}$ and $Q=(q_{j,i})\in \field[x]^{\Delta\times |T|}$.
  The product $P\cdot Q$ can be computed in $\Ot(n^{\omega(s,1-\alpha,t)}\cdot h)=\Ot(n^{\omega(s,1-\alpha,t)+\alpha})$
  time since arithmetic operations on polynomials of degree at most $h$ can be carried out
  in $\Ot(h)$ time~\cite{fft}.

  Now, for $i\in S$, $j\in T$, and $k\in \{h,\ldots,2h-1\}$ consider the coefficient $d_{i,j,k}$
  of $x^k$ in the polynomial $(P\cdot Q)_{i,j}$
  of degree at most $3h$. We have:
  \begin{equation*}
    d_{i,j,k}=\sum_{t=0}^{\Delta-1} \sum_{l=1}^{h} r_{i,t\cdot h+(k-h)+l}\cdot g_{t\cdot h+l,j}.
  \end{equation*}
  Now consider the element $A_{i,j}^k$, for $k\in[h]$:
  \begin{align*}
    A_{i,j}^k&=(UC^kU^{-1})_{i,j}=\sum_{z=1}^n (UC^k)_{i,z}\cdot g_{z,j}\\
    &=\sum_{z=1}^n r_{i,z+k-1}\cdot g_{z,j}=\sum_{t=0}^{\Delta-1}\sum_{l=1}^h r_{i,t\cdot h+(k-1)+l}\cdot g_{t\cdot h+l,j}=d_{i,j,k+h-1}.
  \end{align*}
  We conclude that all the required entries in the respective submatrices $S\times T$ of $A^1,\ldots,A^h$ are encoded as coefficients
  of the (polynomial) entries of the matrix $P\cdot Q$.
\end{proof}
Note that Lemma~\ref{l:rep} together with Theorem~\ref{t:powers} imply Theorem~\ref{t:frobenius-powers}.

\section{Maintaining an FNF under generic rank-1 updates}\label{s:rank1}
Let $A\in \field^{n\times n}$ be again a generic matrix. In this section, we consider the following problem.
Let a Frobenius normal form
of $A$ be given. Suppose $A$ is subject
to a \emph{rank-1 update}, i.e., $A$ is replaced with $A'=A+ab^T$ for some $a,b\in \field^{n\times 1}$.
We require that the update is also \emph{generic}, i.e., the obtained matrix $A'$ is also generic.
We would like to recompute an FNF
of the updated matrix $A'$ faster than from scratch which would
take $\Ot(n^{\omega})$ time.
In this section, we show:
\trankone*
Let us remark that the assumption that the FNFs of both $A$ and its transpose $A^T$ are maintained is merely
for simplicity of exposition. Recall that $A^T$ is generic if and only if $A$ is generic.

We now describe the update procedure.
Our goal is to compute, for some vectors $u,v\in \field^{n\times 1}$ chosen randomly, the $2n$ vectors
$(A+ab^T)^i\cdot u$, and $v^T\cdot (A+ab^T)^i$, for $i=0,\ldots,n-1$.
By Theorem~\ref{t:random-vector}~and~Lemma~\ref{l:eberly}, this will give an explicit FNF of $A+ab^T$
using $\Ot(n^2)$ additional time with probability $\Omega(1/\log{n})$. 
By applying the same procedure to $A^T$, we will compute an FNF of $(A')^T=A^T+ba^T$ as well.
After trying $O(\polylog{n})$ times, we will succeed with high probability.

In the following, will only focus on computing all $(A+ab^T)^i\cdot u$, since all $v^T\cdot (A+ab^T)^i$ can be
computed by proceeding symmetrically with the transpose $A^T$.

First of all, using Lemmas~\ref{l:rep}~and~\ref{l:powers-vector}
applied to $A$ and its FNF, we compute
the vectors $\delta_i:=A^iu$, for $i=0,\ldots,n-1$
in $\Ot(n^2)$ time.
Similarly, we compute
the vectors $\alpha_i:=A^ia$, for $i=0,\ldots,n-1$
within the same time bound.

Note that for any $k=0,\ldots,n-1$, we can expand $X_k:=(A+ab^T)^ku$ as follows:
\begin{align*}
  X_k=(A+ab^T)^ku&=A^ku+\sum_{l=1}^k \left(A^{l-1}\cdot ab^T\cdot (A+ab^T)^{k-l}\cdot u\right)\\
  &=A^ku+\sum_{l=1}^k \left((A^{l-1}a)\cdot (b^T\cdot (A+ab^T)^{k-l}\cdot u)\right)\\
  &=\delta_k+\sum_{l=0}^{k-1} \alpha_{l}\cdot (b^T\cdot X_{(k-1)-l})\\
  &=\delta_k+\sum_{l=0}^{k-1} \alpha_{(k-1)-l}\cdot (b^T\cdot X_{l}).
\end{align*}
This way we obtain a dynamic programming formula for computing the desired subsequent
vectors $X_0,\ldots,X_{n-1}$. The right-hand side of the recurrence involves
a convolution of scalars (obtained from the previous terms) with the precomputed vectors $\alpha_l$.
Recurrences of this kind can be evaluated efficiently
using a folklore combination of FFT and a divide-and-conquer approach, as follows.

Let us initialize vectors $X'_0,\ldots,X'_{n-1}\in \field^{n\times n}$ whose purpose
is to store partially computed vectors $X_0,\ldots,X_{n-1}$.
Initially, put $X'_i=\delta_i$.
Let us define a recursive procedure $F(p,q)$, for $0\leq p\leq q\leq n-1$ such
that, assuming that for all $j=p,\ldots,q$ we have:
\begin{equation}\label{eq:prereq}
  X_j'=\delta_j+\sum_{l=0}^{p-1}\alpha_{(j-1)-l}\cdot (b^T\cdot X_l),
\end{equation}
updates $X_p',\ldots,X_q'$ so that $X'_j=X_j$ for all $j=p,\ldots,q$.
Observe that the prerequisite of $F(0,n-1)$ is satisfied by the initial
values of $X'_i$. Moreover, after $F(0,n-1)$ completes we are done.

Let us now discuss how $F(p,q)$ is implemented. If $p=q$, there is nothing
to be done, as the prerequisite~\eqref{eq:prereq} already implies
that $X_p'=X_p$.
Suppose $p<q$ and let $m=\lfloor (p+q)/2\rfloor$.
We have $p\leq m<q$.
First, $F(p,m)$ is called; note that the prerequisite~\eqref{eq:prereq} of that call holds.
Afterwards, each vector $X_l'$, for $l\in \{p,\ldots,m\}$, equals the respective
vector $X_l'$.
The next step is to compute the scalars $d_l:=b^T\cdot X_l$ for all such $l$
in $O((m-p)n)=O((q-p)n)$ time.

Subsequently, define the following polynomials $\in\field[x]$ of degree at most $q-p$:
\begin{align*}
  Q(x)&=\sum_{i=0}^{m-p}d_{p+i}\cdot x^i,\\
  P_k(x)&=\sum_{i=0}^{q-p}\alpha_{i,k}\cdot x^i,
\end{align*}
where $\alpha_{i,k}$, for $k=1,\ldots,n$, is the $k$-th coordinate of $\alpha_i$.
For each $k=1,\ldots,n$, compute the polynomial $P_k(x)Q(x)$
using FFT~\cite{fft} in $\Ot(q-p)$ time. Through all $k$, this takes $\Ot((q-p)n)$ time. Observe
that the coefficient $c_{k,i}$ of $x^i$ in $P_k(x)Q(x)$ equals
\begin{equation*}
  \sum_{l=\max(0,i-q+p)}^{\min(i,m-p)} \alpha_{i-l,k}\cdot d_{p+l}=\sum_{l=\max(p,i-q+2p)}^{\min(p+i,m)}\alpha_{i+p-l,k}\cdot d_l.
\end{equation*}
In particular, for $j=m+1,\ldots,q$, $j-p-1\in [0,q-p-1]$.
So we get $(j-p-1)+p=j-1\geq m$ and $(j-p-1)-q+2p\leq p-1$.
Thus, the coefficient of $x^{j-p-1}$ in $P_k(x)Q(x)$
equals
\begin{equation*}
  c_{k,j-p-1}=\sum_{l=p}^{m}\alpha_{(j-1)-l,k}\cdot d_l.
\end{equation*}
We conclude that the following column vectors can be retrieved from the computed polynomials:
\begin{align*}
  \Delta_j:=
  {
  \begin{bmatrix}
    c_{1,j-p-1} &
    c_{2,j-p-1} &
    \ldots &
    c_{n,j-p-1}
  \end{bmatrix}}^T=\sum_{l=p}^{m} \alpha_{(j-1)-l}\cdot (b^T\cdot X_l).
\end{align*}
After adding $\Delta_j$ to $X'_j$ for each $j=m+1,\ldots,q$,
we have
\begin{equation*}
  X_j'=\delta_j+\sum_{l=0}^{p-1} \alpha_{(j-1)-l}\cdot (b^T\cdot X_l)+\sum_{l=p}^m\alpha_{(j-1)-l}\cdot (b^T\cdot X_l)=\delta_j+\sum_{l=0}^m \alpha_{(j-1)-l}\cdot (b^T\cdot X_l).
\end{equation*}
As a result, the prerequisite of the call $F(m+1,q)$ is satisfied
and we can call $F(m+1,q)$ to update $X'_{m+1},\ldots,X'_q$
so that they store $X_{m+1},\ldots,X_q$ respectively.
The correctness of the procedure $F(p,q)$ follows easily by induction on $q-p$.

The time $T(N)$ needed to compute $F(l,r)$ when $r-l=N$ clearly
satisfies:
\begin{align*}
  T(1)&=O(1)\\
  T(N)&\leq T(\lceil N/2 \rceil)+T(\lfloor N/2 \rfloor)+\Ot(Nn).
\end{align*}
Therefore, we get $T(N)=\Ot(Nn)$. Since the root call $F(0,n-1)$ satisfies $r-l=n$,
all the desired vectors $X_k=(A+ab^T)^k$ for $k=0,\ldots,n-1$ are
computed in $\Ot(n^2)$ time as desired.
\section{Accessing matrix powers under batch element updates}\label{s:inverse}
This section is devoted to proving the following lemma on computing entries
of the powers of a matrix $A$ under single-element updates to $A$. This result is implicit in the
known works on dynamic matrix inverse~\cite{BrandNS19, Sankowski04, Sankowski05}.

\begin{lemma}\label{l:matrix-inverse}
  Let $A\in \field^{n\times n}$. Let $\Psi=\{(u_1,v_1,y_1),\ldots,(u_f,v_f,y_f)\}\subseteq [n]\times [n]\times \field$
  be such that all pairs $(u_i,v_i)$ are distinct.
  Let $S=\{u_1,\ldots,u_f\}$ and $T=\{v_1,\ldots,v_f\}$.
  Suppose the matrix $B$ is obtained from $A$ by setting $A_{u_i,v_i}:=y_i$ for all $i=1,\ldots,f$.
  Let $h\geq 1$ be an integer.

  Given the submatrices $(A^1)_{T,S},\ldots,(A^h)_{T,S}$, one can preprocess $\Psi$
  in $\Ot(f^\omega\cdot h)$ time, so that the following queries are supported.
  Given any $X,Y\subseteq [n]$ along with the submatrices\linebreak
  $(A^1)_{X,S},\ldots,(A^h)_{X,S}$,
  $(A^1)_{T,Y},\ldots,(A^h)_{T,Y}$, and $(A^1)_{X,Y},\ldots,(A^h)_{X,Y}$,
  compute the submatrices $(B^1)_{X,Y},\ldots,(B^h)_{X,Y}$.
  The query time is $\Ot((\mm(|X|,f,|Y|)+\mm(f,f,\min(|X|,|Y|))\cdot h)$.
\end{lemma}

  Consider the ring
  $\field[X]/(X^{h+1})$ of polynomials in $X$ over $\field$ modulo $X^{h+1}$.
  Consider a matrix polynomial $I-X\cdot A$, where $X$ is a variable.
  $I-X\cdot A$ can be also viewed as a matrix of degree $\leq 1$ polynomials in $X$.
  The inverse $(I-X\cdot A)^{-1}$
  exists in the ring $\left(\field[X]/(X^{h+1})\right)^{n\times n}$ of polynomial matrices modulo $X^{h+1}$
  and equals:
  \begin{equation}
    (I-X\cdot A)^{-1}=I+X\cdot A+X^2\cdot A^2+\ldots +X^h\cdot A^h.
  \end{equation}
  The above equation can be seen to hold by multiplying both sides (modulo $X^{h+1}$) by $I-X\cdot A$.
  
  For each $i=1,\ldots,f$, let $\delta_i:=y_i-A_{u_i,v_i}$, so that $B$ can be
  seen to be obtained from $A$ \linebreak by adding $\delta_i$ to a corresponding element $A_{u_i,v_i}$.
  Let us put
  \begin{align*}
    U&=\begin{bmatrix} e_{u_1} & e_{u_2} & \ldots & e_{u_f}\end{bmatrix}\in \field^{n\times f},\\
      V&={\begin{bmatrix} e_{v_1} & e_{v_2} & \ldots & e_{v_f}\end{bmatrix}}^T\in \field^{f\times n},\\
        \Delta&=\operatorname{diag}(\delta_1,\ldots,\delta_f)\in \field^{f\times f}.
  \end{align*}
  Then, we have $B=A+U\cdot \Delta\cdot V$. Consider the inverse of $I-X\cdot B$ in $\field[X]/(X^{h+1})$.
  We have:
  \begin{equation*}
    (I-X\cdot B)^{-1}=(I-X\cdot (A+U\Delta V))^{-1}=(I-X\cdot A+U\cdot (-X\Delta)\cdot V)^{-1}.
  \end{equation*}
  Put $Z:=(I-X\cdot A)^{-1}=(z_{i,j})$.
  By the Sherman-Morrison-Woodbury formula~(see, e.g.,~\cite{SMW}), we have:
  \begin{align*}
    (I-X\cdot B)^{-1}&=Z-Z\cdot U\cdot (I+(-X\Delta)VZU)^{-1}\cdot (-X\Delta)V\cdot Z\\
    &=Z-(ZU)\cdot (I-X\Delta (VZU))^{-1}\cdot (-X\Delta)\cdot (VZ).
  \end{align*}
  Above, the matrix $ZU$ ($VZ$) selects the subsequent columns $u_1,\ldots,u_f\in S$ (rows $v_1,\ldots,v_f\in T$, resp.) of $Z$ and arranges them left to right (top to bottom, resp.).
  Similarly, $(VZU)_{i,j}=Z_{v_i,u_j}$ and thus $VZU$ can be read from $Z_{T,S}$ in optimal $O(f^2)$ time.
  
  The identity holds if $(I-X\Delta (VZU))$ is invertible in $\field[X]/(X^{h+1})$. It indeed is, as:
  \begin{equation*}
    (I-X\Delta \cdot (VZU))^{-1}=\sum_{i=0}^h (\Delta VZU)^i \cdot X^i.
  \end{equation*}
  Moreover, since
  \begin{equation*}
    \sum_{i=0}^h (\Delta VZU)^i \cdot X^i=\left(\prod_{j=1}^{\lceil \log{h}\rceil }\left(I+(\Delta XVZU)^{2^j}\right)\right)\bmod{X^{h+1}},
  \end{equation*}
  the matrix $P=(I-X\Delta (VZU))^{-1}$ can be computed, given $Z_{T,S}$, using $\Ot(1)$ multiplications of $f\times f$ matrices
  whose entries are polynomials of degree at most $h$ (and the polynomial arithmetic is performed modulo $X^{h+1}$),
  that is, in $\Ot(f^\omega\cdot h)$ time. Recall that since $Z=I+A+\ldots+A^h$,
  the polynomial matrix $Z_{T,S}$ is encoded using the given submatrices $(A^1)_{T,S},\ldots,(A^h)_{T,S}\in \field^{f\times f}$.

  Assuming the matrix $P$ is precomputed, given $X,Y\subseteq [n]$, we can compute
  $(I-X\cdot B)^{-1}_{X,Y}$ in $\Ot(f^2\cdot h)$ time using (rectangular) matrix multiplication as:
  \begin{equation*}
    (I-X\cdot B)^{-1}_{X,Y}=Z_{X,Y}-(ZU)_{X,[f]}\cdot P \cdot (-X\Delta)\cdot (VZ)_{[f],Y}.
  \end{equation*}
  Here, the matrixces $(ZU)_{X,[f]}$ and $(VZ)_{[f],Y}$ can be read from the submatrices $Z_{X,S}$, $Z_{T,Y}$
  respectively.
  The used order of multiplication
  depends on which of $X,Y$ is smaller.
  This expression can be evaluated as
  $((((ZU)_{X,[f]}\cdot P) \cdot (-X\Delta))\cdot (VZ)_{[f],Y}$ in $\Ot((\mm(|X|,f,f))+\mm(|X|,f,|Y|))\cdot h)$
  time, or as
  $(ZU)_{X,[f]}\cdot (P \cdot ((-X\Delta)\cdot (VZ)_{[f],Y}))$
  $\Ot((\mm(f,f,|Y|))+\mm(|X|,f,|Y|))\cdot h)$ time.

  Finally, recall that the submatrices $Z_{X,S}$, $Z_{T,Y}$, and $Z_{X,Y}$ are encoded by the
  submatrices $(A^1)_{X,S},\ldots,(A^h)_{X,S}$,
  $(A^1)_{T,Y},\ldots,(A^h)_{T,Y}$, and $(A^1)_{X,Y},\ldots,(A^h)_{X,Y}$, respectively.
  The desired submatrices $(B^1)_{X,Y},\ldots,(B^h)_{X,Y}$ are encoded by the
  obtained polynomial matrix $(I-X\cdot B)^{-1}_{X,Y}$.

\section{Representing the graph with a generic matrix}
Let $G=(V, E)$ be an unweighted digraph. Let $\xvarset$ be a set of variables
$\xvar_{u,v}$ indexed with pairs from $V\times V$. Let $\yvarset$ be a set of variables
$\yvar_v$ indexed with $V$.

An $n\times n$ \emph{symbolic adjacency matrix} $\asym(G)\in\field[\xvarset\cup \yvarset]^{n\times n}$ of $G$ is defined
as follows:
\begin{equation*}
  \asym(G)_{u,v}=\begin{cases}
    \xvar_{u,v} \cdot \yvar_v &\text{if }u=v\text{ or }uv\in E,\\
    0 &\text{otherwise}.
  \end{cases}
\end{equation*}
This section is devoted to proving the following:
\begin{theorem}\label{t:graph-matrix}
  For all $u,v\in V$, let us assign $\xvar_{u,v}$ a random element $x_{u,v}\in\field$. Similarly, for all
  $v\in V$ assign $\yvar_v$ a random element $y_v\in \field$. Suppose all these random samples are independent.

  Let $A\in \field^{n\times n}$ be a \emph{weighted adjacency matrix} obtained from $\asym(G)$ using this assignment.
  Then, with probability at least $1-n^4/|\field|$:
  \begin{enumerate}[label=(\arabic*)]
    \item $A$ is generic.
    \item For every $u,v\in V$ and $k\in [n-1]$, $A^k_{u,v}\neq 0$ if and only if there exists a path $u\to v$ in $G$
      of length at most $k$.
  \end{enumerate}
\end{theorem}
We start by referring to the following lemma of Wiedemann~\cite{Wiedemann86}.
\begin{lemma}{\upshape\cite[Section V]{Wiedemann86}}\label{l:wiedemann}
  Let $B\in \field^{n\times n}$ be such that all $n$ leading principal minors of $B$ are non-singular.
  Let $\yvar_1,\ldots,\yvar_n$ be variables.
  Then the discriminant of the characteristic polynomial of $B\cdot \text{diag}(\yvar_1,\ldots,\yvar_n)$ is a non-zero
  polynomial in $\yvar_1,\ldots,\yvar_n$ of degree no more than $2n^3$.
\end{lemma}
With Lemma~\ref{l:wiedemann} in hand,
let us prove that the matrix $A$ is generic.
Let $\asym(G)_{|\xvarset=x}\in \field[\yvarset]^{n\times n}$ be obtained from $\asym(G)$ by setting
all $\xvar_{u,v}$ to $x_{u,v}$.
Let $\asym(G)_{|\yvarset=\mathbf{1}}\in \field[\xvarset]^{n\times n}$ ($\asym(G)_{|\xvarset=x,\yvarset=\mathbf{1}}\in \field^{n\times n}$, resp.) be obtained from $\asym(G)$ ($\asym(G)_{|\xvarset=x}$, resp.) by setting
all $\yvar_v$ to $1$.

\begin{lemma}
  With probability at least $1-n^2/|\field|$, all leading principal minors of $\asym(G)_{|\xvarset=x,\yvarset=\mathbf{1}}$
  are non-singular.
\end{lemma}
\begin{proof}
  Identify $V$ with $[n]$.
  Consider a $k\times k$ leading principal minor
  $\asym_k(G)_{|\yvarset=\mathbf{1}}$
  of  $\asym(G)_{|\yvarset=\mathbf{1}}$.
  $\det(\asym_k(G)_{|\yvarset=\mathbf{1}})$ is a polynomial of degree $k$ in $\xvarset$
  containing a monomial $\prod_{i=1}^k\xvar_{i,i}$, and thus is not a zero polynomial.
  Since the assignment $\xvarset=x$ is random, by the Schwartz-Zippel lemma~\cite{Schwartz80,Zippel79},
  $\det(\asym_k(G)_{|\xvarset=x,\yvarset=\mathbf{1}})\neq 0$ with probability at least $1-k/|\field|\geq 1-n/|\field|$.
  By the union bound, the probability that all the $n$ leading principal minors
  are non-singular is at least $1-n^2/|\field|$.
\end{proof}
The following corollary proves item~(1) of Theorem~\ref{t:graph-matrix}.
\begin{corollary}
  With probability at least $1-n^4/|\field|$, $A$ is generic.
\end{corollary}
\begin{proof}
  We simply apply Lemma~\ref{l:wiedemann} to the matrix $\asym(G)_{|\xvarset=x,\yvarset=\mathbf{1}}$
  and hence obtain that the discriminant of the characteristic polynomial of
  $\asym(G)_{|\xvarset=x}=\asym(G)_{|\xvarset=x,\yvarset=\mathbf{1}}\cdot \text{diag}(\yvar_1,\ldots,\yvar_n)$
  is a non-zero polynomial of degree no more than $2n^3$.
  As a result, if one randomly assigns field elements to $\yvar_1,\ldots,\yvar_n$, the
  discriminant of the characteristic polynomial $p_A(t)$ of $A$
  is non-zero with probability at least $1-2n^3/|\field|\geq 1-n^4/|\field|$
  by the Schwartz-Zippel lemma.
  Equivalently, $p_A(t)$ has $n$ distinct roots in an algebraically
  closed extension $\bar{\field}$ of $\field$.
  As the minimal and characteristic polynomials of a matrix have the same roots,
  we obtain that $p_A\equiv \mu_A$ when $A$ is seen as an $n\times n$ matrix over~$\bar{\field}$.
  But it is well-known that neither the characteristic- nor the minimal
  polynomial of a matrix depends on the base field, and consequently
  $p_A\equiv \mu_A$ even if $A$ is seen as a matrix over $\field$.
\end{proof}
Let us now move to proving item~(2) of Theorem~\ref{t:graph-matrix}.
\begin{lemma}\label{l:entry-prob}
  Let $u,v\in V$ and $k\in [n-1]$. Then, with probability at least $1-2k/|\field|$,
  $A_{u,v}^k\neq 0$ if and only if there exists a $u\to v$ path of length
  no more than $k$ in $G$.
\end{lemma}
\begin{proof}
  From the definition of matrix multiplication, one can easily prove inductively that:
  \begin{equation}\label{eq:prod-sum}
    \asym(G)_{u,v}^k=\sum_{\substack{(u_1,\ldots,u_{k+1})\in V^{k+1}\\u_1=u\\u_{k+1}=v}} \left(\prod_{i=1}^k \asym(G)_{u_i,u_{i+1}}\right).
  \end{equation}
  Observe that for a given $(u_1,\ldots,u_{k+1})$ in the sum above, by the definition of $\asym(G)$, the product $\prod_{i=1}^k \asym(G)_{u_i,u_{i+1}}$
  is a non-zero monomial iff for all $i=1,\ldots,k$, either $u_i=u_{i+1}$ or $u_iu_{i+1}\in E$.

  Suppose $\asym(G)_{u,v}^k$ is a non-zero polynomial. Then, the sum~\eqref{eq:prod-sum} contains
  at least one non-zero monomial $\prod_{i=1}^k \asym(G)_{u_i,u_{i+1}}$ corresponding
  to a $(k+1)$-tuple $(u_1,\ldots,u_{k+1})$ with $u_1=u$ and $u_{k+1}=v$.
  If one assumes that $G$ contains self-loops, then $\prod_{i=1}^k \asym(G)_{u_i,u_{i+1}}$
  certifies the existence of a $u\to v$ path consisting of $k$ edges or self-loops in $G$.
  By eliminating the self-loops, one obtains that there exists a $u\to v$ path in $G$
  with at most $k$ edges.

  Now suppose that some shortest $u\to v$ path $P$ in $G$ has length $l\leq k$.
  We prove that  $\asym(G)_{u,v}^k$ is a non-zero polynomial in that case.
  Let $P=e_1\ldots e_l$, where $u_iv_i=e_i\in E$. Clearly, $u_i\neq v_i$ since~$P$ is shortest.
  Set $u_{j}:=v$ for all $j=l+1,\ldots,k+1$.
  Note that all $u_1,\ldots,u_{l+1}$
  are distinct since~$P$ cannot contain cycles.
  Consider the monomial $M=\left(\prod_{i=1}^l x_{u_i,u_{i+1}}\cdot y_{u_{i+1}}\right) \cdot (x_{v,v}y_v)^{k-l}$.
  We now argue that this monomial $M$ appears in the sum~\eqref{eq:prod-sum} precisely
  once, contributed by the tuple $(u_1,\ldots,u_{k+1})$.
  For contradiction, suppose there exists some other
  tuple $(u_1',\ldots,u_{k+1}')$ contributing the same monomial $M$.
  Let $j>1$ be the first index such that $u_j\neq u_j'$.

  If we had $u_j'=u_{j-1}'$, then $u_j\neq u_{j-1}$ since $u_{j-1}=u'_{j-1}$.
  As a result, $u_{j-1}'=u_{j-1}\neq v$.
  Consequently, the monomial contributed by $(u_1',\ldots,u_{k+1}')$
  contains a variable $x_{u_{j'},u_{j'}}\neq x_{v,v}$ that $M$ does not contain,
  a contradiction.

  Therefore, $u_j'\neq u_{j-1}'$. Then, the monomial contributed by
  $(u_1',\ldots,u_{k+1}')$ contains the variable $x_{u_{j-1},u_j'}$.
  $M$ contains only a single variable of the form $x_{u_{j-1},\cdot}$,
  namely $x_{u_{j-1},u_j}$. But $u_j\neq u'_j$, a contradiction.
  Therefore, $M$ is indeed a monomial of $\asym(G)_{u,v}^k$ and thus $\asym(G)_{u,v}^k$ is non-zero.

  We conclude that $\asym(G)_{u,v}^k$ is a non-zero polynomial if and only
  if there exists a $u\to v$ path of length at most $k$ in $G$.
  Finally, $\asym(G)_{u,v}^k\equiv 0$ implies $A_{u,v}^k=0$. On the other hand,
  if $\asym(G)_{u,v}^k\not\equiv 0$, then $\asym(G)_{u,v}^k$ has degree
  at most $2k$, so by the Schwartz-Zippel lemma,
  $A_{u,v}^k\neq 0$ with probability at least $1-2k/|\field|$.
  We obtain that the equivalence is preserved after variable substitution
  with desired probability.
\end{proof}
Item~(2) of Theorem~\ref{t:graph-matrix} follows by combining the above
lemma for all $u,v,k$ via the union bound -- the success probability is
at least $1-n^2\cdot 2\cdot (1+2+\ldots+(n-1))/|\field|\geq 1-n^4/|\field|$.

\subsection{Handling edge weights}\label{s:edge-weights}
Encoding distances via matrix powers crucially requires that all
the edges of $G$ have positive and equal weight, or equivalently,
that $G$ is unweighted.

For a weighted digraph $G=(V,E)$ with $n$ vertices and $m$ edges with integer weights in $[1,W]$, we can,
however, construct a related \emph{unweighted} digraph $G'=(V',E')$ with $nW$ vertices
and $m+n(W-1)$ edges, such that:
\begin{itemize}
  \item Each vertex $v\in V$ corresponds to $W$ vertices $v^1,\ldots,v^W$ in $G'$,
    assembled into a directed path with $W-1$ edges $v^Wv^{W-1},v^{W-1}v^{W-2},\ldots,v^2v^1$.
  \item Each edge $uv\in E$ of weight $c$ has a corresponding edge $u^1v^c$ in $G'$.
\end{itemize}
\begin{lemma}
  For any $u,v\in V$, $\dist_G(u,v)=\dist_{G'}(u^1,v^1)$.
\end{lemma}
\begin{proof}
  If $u=v$, $\dist_G(u,v)=\dist_{G'}(u^1,v^1)=0$. So let us assume $u\neq v$. Then $\dist_G(u,v),\dist_{G'}(u^1,v^1)>0$.

  Let us first prove $\dist_{G'}(u^1,v^1)\leq \dist_G(u,v)$. If $\dist_G(u,v)=\infty$ then this is trivial. Suppose ${\dist_G(u,v)=d}$.
  Then there exists an $u\to v$ path $P=u_1u_2\ldots u_{k}$ of weight $d$ in $G$, where $k\leq d+1$.
  Let $c_i$ be the weight of the edge $u_iu_{i+1}$ in $G$.
  Consider the path
  \begin{equation*}
    P'=(u_1^1u_2^{c_1}\cdot u_2^{c_1}u_2^{c_1-1} \cdot \ldots \cdot u_2^2u_2^1)\cdot (u_2^1u_3^{c_2}\cdot u_3^{c_2}u_2^{c_2-1} \cdot \ldots \cdot u_3^2u_3^1) \cdot \ldots \cdot
    (u_{k-1}^1u_{k}^{c_{k-1}}\cdot u_{k}^{c_{k-1}}u_k^{c_{k-1}-1} \cdot \ldots \cdot u_k^2u_k^1).
  \end{equation*}
  By the construction of $G'$ and the existence of $P$, $P'$ exists in $G'$ and consists of $\sum_{i=1}^{k-1}c_i=d$ edges.
  So indeed $\dist_{G'}(u^1,v^1)\leq \dist_G(u,v)$.

  Now we prove $\dist_G(u,v)\leq \dist_{G'}(u^1,v^1)$. Again, if $\dist_{G'}(u^1,v^1)=\infty$, there is nothing to prove.
  Otherwise, let $\dist_{G'}(u^1,v^1)=d\geq 1$. There exists an $u^1\to v^1$ path $Q=z_1^{p_1}\ldots z_{d+1}^{p_{d+1}}$ in $G'$,
  where $z_1,\ldots,z_{d+1}\in V$, $z_1=u$, $p_1=1$, $z_{d+1}=v$, and $p_{d+1}=1$.
  Let $j_1<\ldots<j_k$ be all indices $j$ such that $p_j=1$. In particular, $j_1=1$ and $j_k=d+1$.
  Since in $G'$, a vertex $w^q\in V'$, for $q>1$ has only a single outgoing edge $w^qw^{q-1}$,
  $j_i<l<j_{i+1}$ implies that $z_l=z_{l-1}$ and $p_l=p_{l-1}-1$. As a result,
  for $i>1$, $z_{j_i}=z_{j_{i-1}+1}$.
  We obtain that $P'$ can be expressed as:
  \begin{equation*}
    P'=(z_1^{1}z_2^{p_2}z_2^{p_2-1}\ldots z_2^1)\cdot (z_{j_2}^1z_{j_2+1}^{p_{j_2+1}} z_{j_2+1}^{p_{j_2+1}-1}\ldots z_{j_2+1}^{1})\cdot \ldots\cdot
    (z_{j_{k-1}}^1z_{j_{k-1}+1}^{p_{j_{k-1}+1}} z_{j_{k-1}+1}^{p_{j_{k-1}+1}-1}\ldots z_{j_{k-1}+1}^{1}).
  \end{equation*}
  Hence we conclude that $P'$ has $\sum_{i=1}^{k-1}p_{j_i+1}$ edges.
  Moreover, by the construction of $G'$, for each $i=1,\ldots,k-1$, there exists an edge $z_{j_i}z_{j_i+1}$ of weight
  $p_{j_i+1}$ in $G$. As a result, there exists a path $z_{j_1}\to z_{j_k}=z_1\to z_{d+1}=u\to v$ of
  weight $\sum_{i=1}^{k-1}p_{j_i+1}=|P'|=d$. This implies the desired inequality $\dist_G(u,v)\leq \dist_{G'}(u^1,v^1)$.
\end{proof}
Finally, let us note that by the correspondence of edges in $G$ and $G'$, an insertion or deletion (failure)
or a single edge $uv$ of weight $c$ in $G$ can be reflected by a single edge insertion or deletion of the edge $u^1v^c$ in $G'$.
Similarly, a failure of a vertex $v$ in $G$ can be translated to a failure of a single vertex $v^1$ in $G'$.

  \section{Distance sensitivity oracles}\label{s:dso}

\newcommand{\tin}{\text{in}}
  \newcommand{\tout}{\text{out}}

  Let $G=(V, E)$ be a digraph.
  Recall that a distance sensitivity oracle (DSO) is a data structure answering queries about
  $\dist_{G-F}(s,t)$, where $s,t\in V$ and $F\subseteq V\cup E$.
  The DSO problem can also be generalized by introducing the \emph{update} procedure that takes the set
  $F$ and preprocesses the failures to support efficient queries $(s,t)$ about $\dist_{G-F}(s,t)$
  with the failures $F$ fixed. Such a variant has been studied, e.g., in~\cite{BrandS19} and the objective
  is to give a tradeoff between the preprocessing, update, and query times.

  In the following, we will focus, without loss of generality, on edge failures only. In directed graphs, vertex failures can
  be easily reduced to edge failures via a standard vertex-splitting trick, as described next.
  Construct a related graph $G'$, at most twice as large as $G$, as follows.
  Each vertex $v$ of $G$ is split into two vertices $v_\tin,v_\tout$ connected
  by a directed edge $v_\tin v_\tout$.
  Each edge $uv\in E$ gives rise to an edge $u_\tout v_\tin$ in $G'$.
  Every $k$-edge $s\to t$ path $P$ in $G$ corresponds to a $2k+1$-edge path
  $P'=s_\tin\to t_\tout$ in $G'$ such that~$P$ goes through a vertex $z$ in $G$
  iff $P'$ goes through the edge $z_\tin z_\tout$ in $G'$.
  As a result, a failure of vertex $z$ of $G$ can be simulated using a failure
  of the edge $z_\tin z_\tout$ in $G'$.
  Clearly, if $F'\subseteq E(G')$ is obtained from $F\in V(G)\cup E(G)$ by replacing
  failing vertices with failing edges this way, then
  $\dist_{G-F}(s,t)=(\dist_{G'-F'}(s_\tin,t_\tout)-1)/2$.

  \subsection{Single failures}
  Let us first consider the 1-DSO problem, i.e., we only allow queries of the form
  $(s,t,F)$, where $F$ contains a single edge $f$ of $G$.
    \cite{GuR21, Ren22}~showed the following reduction of the 1-DSO problem to the \emph{$h$-truncated} 1-DSO
  problem where one is only interested in supporting queries computing $\min(\dist_{G-f}(s,t),h)$ instead.
  \begin{theorem}\label{t:truncated-reduction}{\upshape \cite[Section~3.3]{GuR21}}
    Let $G$ be an unweighted digraph. Suppose an $h$-truncated \linebreak 1-DSO $\mathcal{D}_h$ for $G$ with preprocessing time $P(n)$
    and query time $Q(n)$ is given. Then a general Monte-Carlo randomized 1-DSO $\mathcal{D}$ for~$G$ with $O(1)$ query time and $\Ot(n^2)$ space
    can be constructed in  $\Ot(P(n)+n^{2+\rho}+n^2\cdot Q(n)+n^3/h)$ time. If $\mathcal{D}_h$ produces correct answers w.h.p.,
    then so does~$\mathcal{D}$.
  \end{theorem}

  \cite{GuR21} showed an $h$-truncated 1-DSO with $\Ot\left(n^\omega+n^{\omega(1,1-\alpha,1-\alpha)+2\alpha}\right)$ preprocessing time for
  $h=\Theta(n^{\alpha})$. This implies, by Theorem~\ref{t:truncated-reduction},
  a general 1-DSO with preprocessing time $O(n^{2.58})$ and $O(1)$ query time if $h$ is chosen appropriately.
  This construction time bound does not, however, match the $\Ot(n^{2+\rho})=\Ot(n^{2.529})$ time bound
  of Zwick's APSP algorithm~\cite{Zwick02}.
  The $h$-truncated 1-DSO (and also the general 1-DSO) of~\cite{GuR21} also generalizes to digraphs with integer weights $[1,W]$
  at the cost of an additional factor $W$ in the preprocessing time.
  We give an improved $h$-truncated 1-DSO for unweighted digraphs, as captured by the following lemma.
  \begin{lemma}\label{l:our-truncated}
    Let $G$ be an unweighted digraph. For $h=\Theta(n^{\alpha})$, there exists an
    $h$-truncated DSO with $\Ot(n^{\omega(1,1-\alpha,1)+\alpha})$ preprocessing
    time, $O(n^2h)$ space and $\Ot(h)$ query time.
    The data structure is Monte Carlo randomized and answers queries
    correctly with high probability.
  \end{lemma}
  \begin{proof}
  Fix the field $\field$ to be $\mathbb{Z}/p\mathbb{Z}$ for some prime number
    $p=\Theta(n^{4+c})$, where $c\geq 1$ is a constant controlling the error probability.
    Let $A\in \field^{n\times n}$ be a weighted adjacency matrix of Theorem~\ref{t:graph-matrix}.
    Recall that $A$ is generic with probability at least $1-1/n^c$.
  The preprocessing is simply to compute the matrix powers $A^1,\ldots,A^{h}$,
    which can be done in $\Ot(n^{\omega(1,1-\alpha,1)+\alpha})$ time by Theorem~\ref{t:powers}.
    Theorem~\ref{t:powers} requires an FNF of $A$, which can be computed in $\Ot(n^\omega)$ time by Lemma~\ref{l:fnf-simple},
    and an auxiliary matrix $R$ of Lemma~\ref{l:rep}, computed in $\Ot(n^2)$ time.
    Note that $\omega(1,1-\alpha,1)+\alpha\geq\omega$.

  Observe that if the graph $G$ is subject to a failure of a single edge $f=uv$, the weighted adjacency matrix $A$ -- assuming the same variable substitution in $\asym(G)$ -- undergoes
  a single \emph{element update} of changing the entry $A_{u,v}$, $u\neq v$, to $0$.
    Let $B$ denote the matrix $A$ after such an update.
    By Lemma~\ref{l:matrix-inverse}, for any $s,t\in V$, we can
    compute $(B^1)_{s,t},\ldots,(B^{h})_{s,t}$ in $\Ot(h)$ time.
    By Theorem~\ref{t:graph-matrix}, with high probability, if $d\leq h$ is minimal such that $(B^d)_{s,t}\neq 0$, then
    $\dist_{G-f}(s,t)=d$, and otherwise, if such a value $d$ does not exists then $\dist_{G-f}(s,t)>h$
    and thus $\min(\dist_{G-F}(s,t),h)=h$.
  \end{proof}

  By using the above lemma for $h=\Theta(n^{1-\rho})$ and applying Theorem~\ref{t:truncated-reduction}, we have\footnote{Similarly as in~\cite{GuR21, Ren22}, the size of the obtained DSO is $\Ot(n^2)$ even though superquadratic $\Ot(n^{3-\rho})$ space is used during the construction phase.}:
  {\renewcommand{\footnote}[1]{}\ouronedso*}

  Notably, the obtained data structure of Theorem~\ref{t:our-1-dso} matches Zwick's best-known APSP\linebreak
  bound~\cite{Zwick02}
  in terms of preprocessing time (up to polylog factors) and has optimal $O(1)$ query time.
  As discussed in Section~\ref{s:edge-weights}, the approach can be generalized
  to digraphs with integer weights in $[1,W]$ at the cost of $\Ot(W^{2+\rho})=O(W^{2.529})$ multiplicative overhead
  in the preprocessing time.

  \subsection{Multiple failures}

  Let us now consider supporting an arbitrary number $f$ of edge failures
  in the preprocess-update-query model. We will show the following.

    \ourmanydso*

  \paragraph{Hitting sets.} Before we continue, let us recall a standard \emph{hitting set trick}~\cite{UY91} that proved useful in solving shortest path problems across multiple settings in the past.

  \begin{lemma}\label{l:hitting}
    Let $G$ be an unweighted digraph. Let $h\in [1,n]$ be an integer.
    Let $H\subseteq V$ be a subset of $V$ obtained by sampling $\gamma\cdot (n/h)\log{n}$ vertices uniformly and independently, where $\gamma\geq 1$ is a sufficiently large constant.
    For any $s,t\in V$, let $G_{H,s,t}$ be a weighted digraph
    on $H\cup \{s,t\}$ such that for any $u,v\in V(G_{H,s,t})$,
    an edge $uv$ of weight $\dist_G(u,v)$ appears in $E(G_{H,s,t})$
    iff $\dist_G(u,v)\leq h$. Then, $\dist_G(s,t)=\dist_{G_{H,s,t}}(s,t)$
    holds with high probability depending on the constant $\gamma$.
  \end{lemma}
  Lemma~\ref{l:hitting} reduces computing $\dist_G(s,t)$ to finding
  $h$-bounded distances between $\Ot(n/h)$ vertices of $G$.
  Once the (potentially dense) auxiliary graph $G_{H,s,t}$ from Lemma~\ref{l:hitting} is constructed,
  obtaining the desired $s,t$-distance amounts to running Dijkstra's algorithm
  on $G_{H,s,t}$ in $\Ot((n/h)^2)$ time.
  Significantly, a sampled hitting set $H$ is valid for any graph on $V$, i.e.,
  with high probability, the same $H\subseteq V$ can be used with $\poly(n)$ (possibly random) different
  graphs, as long as $H$ is independent of these graphs.
  For example, in the dynamic setting, $H$ is valid for $\poly(n)$ versions
  of the evolving graph $G$ if the queries do not reveal any information about $H$.
  In particular, for the studied oracles computing \emph{exact} distances,
  the answers are uniquely determined by the input graph and thus do not
  reveal the random bits behind the used hitting sets.

  \paragraph{Preprocessing.} The only preprocessing is to construct a weighted adjacency
  matrix $A$ (over a sufficiently large $\field$ for $G$ as described in Theorem~\ref{t:graph-matrix}) and its FNF along with
  the auxiliary matrix~$R$ from~Lemma~\ref{l:rep}, which costs $O(n^2)$ space.
  The preprocessing takes $\Ot(n^\omega)$ time by Lemma~\ref{l:fnf-simple}.

  \paragraph{Update.} Given a batch $F$ of $f=\Theta(n^\beta)$ failing edges, $\beta<1$, we proceed as follows.
  Let $S$ contain all the endpoints of the failing edges $F$. We have $|S|=O(f)$.
  Let $H$ be a sampled hitting set from Lemma~\ref{l:hitting} for $h=\lceil n/f\rceil$.
  As a result, $|H|=\Ot(n/h)=\Ot(f)$.
  Using Theorem~\ref{t:powers}, we compute the submatrices
  $(A^1)_{S\cup H},\ldots,(A^h)_{S\cup H}$ in $\Ot(n^{\omega(\beta,\beta,\beta)+(1-\beta)})=\Ot(n^{1+\beta\cdot (\omega-1)})=\Ot(nf^{\omega-1})$ time.

  Consider the weighted adjacency matrix $B$ of $G-F$. $B$ is obtained from $A$ by zeroing
  the entries $A_{u,v}$ for all $uv\in F$. Therefore, $B$ is obtained from $A$ via $f$ element updates.
  By Lemma~\ref{l:matrix-inverse}, for any $x,y\in H$, the elements
  $(B^1)_{x,y},\ldots,(B^h)_{x,y}$ can be computed, given the preprocessed submatrices of the powers of $A$,
  in $\Ot(f^\omega \cdot h)=\Ot(nf^{\omega-1})$ time.
  Recall that by Theorem~\ref{t:graph-matrix}, the submatrices  $(B^1)_H,\ldots,(B^h)_H$
  encode the $h$-bounded distance between $H$ in $G-F$.

  The matrices stored upon update use $\Ot(nf)=O(n^2)$ space.

  \paragraph{Query.} Suppose we want to compute $\dist_{G-F}(s,t)$ for query vertices $s,t\in V$.
  We construct the graph $G_{H,s,t}$ from Lemma~\ref{l:hitting}.
  Observe that we have precomputed most of the edges of $G_{H,s,t}$ in the update phase.
  It remains to compute the (weights) of edges incident to $s$ and $t$ in $G_{H,s,t}$.
  To this end, we first compute the submatrices $(A^1)_{\{s,t\},S\cup H\cup \{s,t\}},\ldots,(A^h)_{\{s,t\},S\cup H\cup \{s,t\}}$
  and the submatrices $(A^1)_{S\cup H,\{s,t\}},\ldots,(A^h)_{S\cup H,\{s,t\}}$
  in $\Ot(n^{\omega(0,\beta,\beta)+(1-\beta)})=\Ot(n^{\beta+1})=\Ot(nf)$ time.
  Afterwards, we can apply Lemma~\ref{l:matrix-inverse} to compute
  the submatrices $(B^1)_{\{s,t\},H\cup \{s,t\}},\ldots,(B^h)_{\{s,t\},H\cup \{s,t\}}$
  and the submatrices $(B^1)_{H,\{s,t\}},\ldots,(B^h)_{H,\{s,t\}}$
  encoding the remaining $h$-bounded distances required for constructing
  $G_{H,s,t}$ in $\Ot(n^{\omega(0,\beta,\beta)+(1-\beta)})=\Ot(nf)$ time.
  Having $G_{H,s,t}$ constructed, the final step is to run
  Dijkstra's algorithm to compute the shortest $s,t$ path
  in $G_{H,s,t}$ in $\Ot(f^2)$ time.

   \section{Dynamic distances}\label{s:dyn-distances1}

   In this section, we describe three different distance oracles for fully
   dynamic unweighted digraphs.

   \subsection{Tweaking the data structure of~\cite{BrandFN22}}

  \begin{theorem}\label{t:dynamic-inverse}{\upshape\cite{BrandFN22}}
    Let $B\in \field^{n\times n}$ and let $0\leq \nu\leq \mu\leq 1$. Let $h\in [1,n]$ be an integer.
    Let $S,T\subseteq [n]$.
    There exists a data structure maintaining
    the $S\times T$ submatrix of the
    inverse of the polynomial matrix $I-X\cdot B\in\field[X]/(X^{h+1})$
    under element updates to $B$ and
    single-element changes (additions or removals) to the sets $S$ and $T$ as long as $|S|,|T|\leq n^\mu$.
    The initialization time is $\Ot(h\cdot n^\omega)$ and the worst-case
    update time is
    $\Ot((n^{\omega(1,\mu,1)-\mu}+n^{\omega(1,\nu,\mu)-\nu}+n^{\mu+\nu}+|S|\cdot |T|)\cdot h)$.
  \end{theorem}

  The data structure of Theorem~\ref{t:dynamic-inverse} can be used to maintain
  an unweighted digraph $G$ under single-edge insertions and deletions and answer
  $s,t$-distance queries in $G$ as follows~(see also \cite[Section~C]{dynamic-st-dist}).
  For a parameter $\lfloor n^\alpha\rfloor=h\in [1,n]$, sample a random hitting set $H\subseteq V$
  of size $\widetilde{\Theta}(n/h)$ as in Lemma~\ref{l:hitting}.
  The data structure of Theorem~\ref{t:dynamic-inverse} is set up for the (unweighted)
  adjacency matrix $A^*$ of $G$ and $S=T=H$, and the field $\field$ is chosen
  to be $\mathbb{Z}/p\mathbb{Z}$ for a sufficiently large random prime $p\in n^{\Theta(1)}$.
  As discussed in the proof of Lemma~\ref{l:matrix-inverse}, the maintained $H\times H$ submatrix of
  $(I-XA^*)^{-1}$, encodes the submatrices $((A^*)^1)_{H},\ldots,((A^*)^h)_{H}$.
  Those, in turn, encode, with high probability, the $h$-bounded
  distances between the vertices $H$ in $G$.

  To compute $\dist_G(s,t)$ for query vertices $s,t\in V$, we first temporarily add $s,t$
  to the sets $S$ and~$T$, at the cost of $O(1)$ updates issued to the data structure.
  Afterwards, the maintained submatrix can be used to construct the graph $G_{H,s,t}$
  of Lemma~\ref{l:hitting}, and consequently $\dist_G(s,t)$ can be computed
  in $\Ot((n/h)^2)$ additional time by running Dijkstra's algorithm on $G_{H,s,t}$.
  After $\dist_G(s,t)$ is computed, we remove the temporarily added vertices $\{s,t\}\setminus H$
  from $S$ and $T$.

  Both updates and queries
  are processed in
  $\Ot(n^{\omega(1,\mu,1)-\mu+\alpha}+n^{\omega(1,\nu,\mu)-\nu+\alpha}+n^{\mu+\nu+\alpha}+n^{2-\alpha})$ worst-case time.
  By setting $\mu=0.862$, $\nu=0.543$, and $h=n^{0.297}$, \cite{BrandFN22} obtain
  $\Ot(n^{1.703})$ update and query bound.\footnote{One can use the online term balancer~\cite{Complexity} to reproduce this bound for the given parameters.}

  Using Theorems~\ref{t:graph-matrix}~and~\ref{t:powers}, we can obtain an improved
  bound by slightly altering how the data structure behind Theorem~\ref{t:dynamic-inverse}
  operates when initialized with the \emph{weighted adjacency matrix}~$A$ from Theorem~\ref{t:graph-matrix}
  instead of the unweighted adjacency matrix $A^*$.
  Specifically, the dynamic matrix inverse data structure of Theorem~\ref{t:dynamic-inverse}
  operates, at the topmost level, in phases of $\Theta(n^\mu)$ element updates.
  At the end of each phase, the inverse $(I-XA)^{-1}$ is explicitly recomputed
  from the inverse at the beginning of the phase and the
  $\Theta(n^\mu)$ most recent updates using fast rectangular matrix multiplication
  in $\Ot(n^{\omega(1,\mu,1)}\cdot h)$ time. In a standard way, this cost can be distributed
  over the $\Theta(n^\mu)$ updates and hence the
  $\Ot(n^{\omega(1,\mu,1)-\mu}\cdot h)$ term in the update bound.
  However, we can as well recompute
  $(I-XA)^{-1}\pmod{X^{h+1}}$ \emph{from scratch} using Theorem~\ref{t:powers}
  in $\Ot(n^{\omega(1,1-\alpha,1)+\alpha})$ time as
  \begin{equation*}
    (I-XA)^{-1}\bmod{X^{h+1}}=I+X\cdot A+\ldots+X^h\cdot A^h
  \end{equation*}
  (see Section~\ref{s:inverse}).
  Since this recomputation happens every $\Theta(n^\mu)$ updates, we
  obtain a slightly different
  $\Ot(n^{\omega(1,1-\alpha,1)-\mu+\alpha}+n^{\omega(1,\nu,\mu)-\nu+\alpha}+n^{\mu+\nu+\alpha}+n^{2-\alpha})$
  update/query bound for $h=n^\alpha$, as long as $1-\alpha\leq \mu$.
  By using the online term balancing tool~\cite{Complexity},
  we find that for
  $\mu=0.793$, $\nu=0.552$, and $\alpha=0.328$, the bound is $O(n^{1.673})$.

  \tweaked*

\subsection{Another data structure for single-edge updates}
If $\omega=2$, both our data structure of Theorem~\ref{t:tweaked} and that of~\cite{BrandFN22}
yield an $\Ot(n^{1+5/8})$ update/query bound if properly optimized.
In this section, we show a different dynamic distance oracle
summarized as follows.

\dyndist*

The bound $\Ot\left(n^{\frac{\omega+1}{2}}\right)=O(n^{1.687})$ of Theorem~\ref{t:dyndist} is slightly worse
than the $O(n^{1.673})$ bound obtained in Theorem~\ref{t:tweaked}, but leads to a more natural $O(n^{1.5})$ bound under the optimistic
assumption $\omega=2$. Moreover, it breaks through the theoretical $O(n^{1+5/8})$ limit of the other discussed approaches
already if $\omega<2.25$.

\paragraph{Update.} The algorithm operates in phases of $\lceil n^{1-\alpha}\rceil$ edge updates, for $\alpha\in [0,1]$ to be set later.
At any point of time, we denote by $A$ the weighted adjacency matrix (see Theorem~\ref{t:graph-matrix}) of the graph~$G$ from the beginning
of the current phase, and by $B$
a weighted adjacency matrix of the current graph~$G$.
The matrix~$B$ equals $A$ immediately after the phase starts
and evolves by single-element updates corresponding to the edge updates
issued to $G$. In particular, if an edge $uv$ is inserted into~$G$,
a fresh random field element $x_{u,v}$ is sampled to guarantee
that $B$ is always obtained from $\asym(G)$ via random variable substitution (see Theorem~\ref{t:graph-matrix}).

When a phase starts, we compute in $\Ot(n^\omega)$ time the weighted adjacency matrix $A$ of the graph~$G$
along with an FNF of $A$ (Lemma~\ref{l:fnf-simple}) and the auxiliary matrix $R$ of Lemma~\ref{l:rep}.
This costly computation happens once per phase and thus takes $\Ot(n^{\omega-1+\alpha})$ amortized time per update.
Moreover, for $h=\lceil n^{\alpha} \rceil$
we also sample a hitting set $H\subseteq V$ of size $\Theta(n/h\log{n})=\Ot(n^{1-\alpha})$.

When a phase proceeds, let us denote by $S\subseteq V$ the set of endpoints of
the edges inserted or deleted in the current phase.
At the beginning of a phase, $S=\emptyset$ and we always have $|S|=O(n^{1-\alpha})$.
Throughout, we make sure that all the submatrices
$(A^1)_{H\cup S},\ldots,(A^h)_{H\cup S}$ are stored explicitly.
To this end, when a phase starts, we compute the submatrices
$(A^1)_H,\ldots,(A^h)_H$. This takes $\Ot(n^{\omega(1-\alpha,1-\alpha,1-\alpha)+\alpha})=\Ot(n^{(1-\alpha)\omega+\alpha})\subseteq \Ot(n^{\omega})$ time
by Theorem~\ref{t:powers}.
Amortized over the $\Theta(n^{1-\alpha})$ updates in a phase, this costs
$\Ot(n^{\omega-1+\alpha})$ time.
Upon an update of an edge $uv$, $u$ and $v$ are inserted into $S$,
so we only need to compute the submatrices
$\left((A^j)_{\{u,v\},H\cup S\cup \{u,v\}}\right)_{j=1}^h$ and
$\left((A^j)_{H\cup S,\{u,v\}}\right)_{j=1}^h$
to satisfy the invariant.
This costs $\Ot(n^{\omega(0,1-\alpha,1-\alpha)+\alpha})=\Ot(n^{2-\alpha})$ time by Theorem~\ref{t:powers}.

We also maintain the submatrices $(B^1)_H,\ldots,(B^h)_H$. They are initialized trivially
to the corresponding computed submatrices
$(A^1)_H,\ldots,(A^h)_H$ when a phase starts.
By Lemma~\ref{l:matrix-inverse}, they
can be updated subject to an element change $(u,v)$ (corresponding to an insertion
or deletion of the edge $uv$ in $G$) in $B$ in $\Ot(|H|^2\cdot h)=\Ot(n^{2-\alpha})$ time
if the (current) submatrices $(B^1)_{H\cup\{u,v\}},\ldots,(B^h)_{H\cup\{u,v\}}$ are provided.
To provide those, we only need to construct the submatrices
$\left((B^j)_{\{u,v\},H\cup\{u,v\}}\right)_{j=1}^h$ and
$\left((B^j)_{H,\{u,v\}}\right)_{j=1}^h$,
as the other entries are maintained explicitly.
Again, by Lemma~\ref{l:matrix-inverse}, those can be obtained
from the submatrices $\left((A^j)_{H\cup S\cup \{u,v\}}\right)_{j=1}^h$ (that are off by at most $|S|$ element updates to $A$) in
$\Ot\left(\left(\mm(2,|S|,|H|)+\mm(|S|,|S|,2)\right)\cdot h\right)=$
$\Ot(n^{2-2\alpha}\cdot n^\alpha)=\Ot(n^{2-\alpha})$ time.

\paragraph{Query.} Finally, to answer a distance query, we construct a graph $G_{H,s,t}$
of Lemma~\ref{l:hitting} and run Dijkstra's algorithm
on it in $\Ot(n^{2-2\alpha})$ time. As the $h$-bounded distances between
the vertices~$H$ are all encoded in the maintained
submatrices $(B^1)_H,\ldots,(B^h)_H$, we only need to compute
$h$-bounded distances between $\{s,t\}$ and $H\cup \{s,t\}$.
These, again, can be devised from the submatrices
$\left((B^j)_{\{s,t\},H\cup\{u,v\}}\right)_{j=1}^h$ and
$\left((B^j)_{H,\{s,t\}}\right)_{j=1}^h$.
To construct those, we proceed identically as if an update of the edge $st$
was issued: we can temporarily add $\{s,t\}$ to $S$, recompute the
missing submatrices of the powers of $A$ and $B$, and revert
this process at the end.
This way, constructing the $O(|H|)$ missing edges of $G_{H,s,t}$ takes
$\Ot(n^{2-\alpha})$ time.

The amortized update time of the data structure is $\Ot(n^{2-\alpha}+n^{\omega-1+\alpha})$,
which is optimized for $\alpha=\frac{3-\omega}{2}$.
Observe that the heavy $\Ot(n^\omega)$-time computation, the only
source of amortization here, happens
only when a phase starts, once per $\lceil n^{1-\alpha}\rceil$ updates.
As a result, the amortized bound can be converted into a worst-case
bound using a standard technique, see, e.g.,~\cite{AbrahamCK17, BrandNS19}.

\subsection{Vertex updates}\label{s:vertexupd}
Finally, we show that the dynamic Frobenius form algorithm of Section~\ref{s:rank1} leads
to the first distance oracle supporting distance queries in $\Ot(n)$ time
and \emph{vertex updates} (i.e., changing all the edges incident to a single vertex) significantly faster than $\Ot(n^\omega)$ in the worst-case.
We note that a \emph{static} distance oracle supporting queries in $\Ot(n)$ time
can be constructed in $\Ot(n^\omega)$ time~\cite{YusterZ05}.

\vertexupd*
\begin{proof}
  The data structure is very simple. We maintain a weighted adjacency matrix $A$
  of $G$, as given by Theorem~\ref{t:graph-matrix}.
  We also maintain a Frobenius form of $A$ and $A^T$.
  Since updating all the incoming edges or all the outgoing edges
  of a vertex $v\in V(G)$ can be encoded using a rank-1 update
  of $A$, a vertex update translates to at most $2$ rank-1 updates of $A$.
  Hence, by Theorem~\ref{t:rank1}, the Frobenius forms of $A$ and $A^T$ can be updated
  subject to a vertex update on $G$ in $\Ot(n^2)$ time.
  After each update, we also recompute the auxiliary matrix $R$ of Lemma~\ref{l:rep}
  in $\Ot(n^2)$ time.
  Given an FNF and the auxiliary matrix, for any $s,t\in V$, we can compute
  the entries $(A^1)_{s,t},\ldots,(A^{n-1})_{s,t}$ in $\Ot(n)$ time.
  By Theorem~\ref{t:graph-matrix}, w.h.p., $\dist_{G}(s,t)$
  equals the minimal $d\geq 0$ such that $(A^d)_{s,t}\neq 0$.
\end{proof}

\section*{Acknowledgment}
We would like to thank Maciej Gałązka for important clarifications regarding linear algebra, and anonymous
FOCS reviewers for valuable comments.

\bibliographystyle{alpha}

\bibliography{references}

\newcommand{\etalchar}[1]{$^{#1}$}
\newcommand{\sortkey}[1]{}
\begin{thebibliography}{vdBFN22}

\bibitem[ACK17]{AbrahamCK17}
Ittai Abraham, Shiri Chechik, and Sebastian Krinninger.
\newblock Fully dynamic all-pairs shortest paths with worst-case update-time
  revisited.
\newblock In {\em Proceedings of the Twenty-Eighth Annual {ACM-SIAM} Symposium
  on Discrete Algorithms, {SODA} 2017}, pages 440--452. {SIAM}, 2017.

\bibitem[AIMN91]{AusielloIMN91}
Giorgio Ausiello, Giuseppe~F. Italiano, Alberto Marchetti{-}Spaccamela, and
  Umberto Nanni.
\newblock Incremental algorithms for minimal length paths.
\newblock {\em J. Algorithms}, 12(4):615--638, 1991.

\bibitem[AvdB23]{alokhina-brand}
Anastasiia Alokhina and Jan van~den Brand.
\newblock Fully dynamic shortest path reporting against an adaptive adversary.
\newblock {\em CoRR}, abs/2304.07403, 2023.

\bibitem[BCC{\etalchar{+}}22]{BiloCC0S22}
Davide Bil{\`{o}}, Keerti Choudhary, Sarel Cohen, Tobias Friedrich, and Martin
  Schirneck.
\newblock Deterministic sensitivity oracles for diameter, eccentricities and
  all pairs distances.
\newblock In {\em 49th International Colloquium on Automata, Languages, and
  Programming, {ICALP} 2022}, volume 229 of {\em LIPIcs}, pages 22:1--22:19.
  Schloss Dagstuhl - Leibniz-Zentrum f{\"{u}}r Informatik, 2022.

\bibitem[Ber09]{Bernstein09}
Aaron Bernstein.
\newblock Fully dynamic {(2} + epsilon) approximate all-pairs shortest paths
  with fast query and close to linear update time.
\newblock In {\em 50th Annual {IEEE} Symposium on Foundations of Computer
  Science, {FOCS} 2009}, pages 693--702. {IEEE} Computer Society, 2009.

\bibitem[Ber16]{Bernstein16}
Aaron Bernstein.
\newblock Maintaining shortest paths under deletions in weighted directed
  graphs.
\newblock {\em {SIAM} J. Comput.}, 45(2):548--574, 2016.

\bibitem[BGL03]{BrentGL03}
Richard~P. Brent, Shuhong Gao, and Alan G.~B. Lauder.
\newblock Random krylov spaces over finite fields.
\newblock {\em {SIAM} J. Discret. Math.}, 16(2):276--287, 2003.

\bibitem[BGY80]{BrentGY80}
Richard~P. Brent, Fred~G. Gustavson, and David Y.~Y. Yun.
\newblock Fast solution of toeplitz systems of equations and computation of
  pad{\'{e}} approximants.
\newblock {\em J. Algorithms}, 1(3):259--295, 1980.

\bibitem[BHG{\etalchar{+}}21]{BergamaschiHGWW21}
Thiago Bergamaschi, Monika Henzinger, Maximilian~Probst Gutenberg,
  Virginia~Vassilevska Williams, and Nicole Wein.
\newblock New techniques and fine-grained hardness for dynamic near-additive
  spanners.
\newblock In {\em Proceedings of the 2021 {ACM-SIAM} Symposium on Discrete
  Algorithms, {SODA} 2021}, pages 1836--1855. {SIAM}, 2021.

\bibitem[BHS07]{BaswanaHS07}
Surender Baswana, Ramesh Hariharan, and Sandeep Sen.
\newblock Improved decremental algorithms for maintaining transitive closure
  and all-pairs shortest paths.
\newblock {\em J. Algorithms}, 62(2):74--92, 2007.

\bibitem[BJN94]{basicalgebra}
Phani~Bhushan Bhattacharya, Surender~Kumar Jain, and SR~Nagpaul.
\newblock {\em Basic abstract algebra}.
\newblock Cambridge University Press, 1994.

\bibitem[BK09]{BernsteinK09}
Aaron Bernstein and David~R. Karger.
\newblock A nearly optimal oracle for avoiding failed vertices and edges.
\newblock In {\em Proceedings of the 41st Annual {ACM} Symposium on Theory of
  Computing, {STOC} 2009}, pages 101--110. {ACM}, 2009.

\bibitem[BM21]{BostanM21}
Alin Bostan and Ryuhei Mori.
\newblock A simple and fast algorithm for computing the \emph{N}-th term of a
  linearly recurrent sequence.
\newblock In {\em 4th Symposium on Simplicity in Algorithms, {SOSA} 2021},
  pages 118--132. {SIAM}, 2021.

\bibitem[Bra]{Complexity}
Jan van~den Brand.
\newblock Complexity term balancer.
\newblock \url{www.ocf.berkeley.edu/~vdbrand/complexity/}.
\newblock Tool to balance complexity terms depending on fast matrix
  multiplication.

\bibitem[CC20]{ChechikC20}
Shiri Chechik and Sarel Cohen.
\newblock Distance sensitivity oracles with subcubic preprocessing time and
  fast query time.
\newblock In {\em Proccedings of the 52nd Annual {ACM} {SIGACT} Symposium on
  Theory of Computing, {STOC} 2020}, pages 1375--1388. {ACM}, 2020.

\bibitem[Che15]{Chechik15}
Shiri Chechik.
\newblock Approximate distance oracles with improved bounds.
\newblock In {\em Proceedings of the Forty-Seventh Annual {ACM} on Symposium on
  Theory of Computing, {STOC} 2015}, pages 1--10. {ACM}, 2015.

\bibitem[CLPR12]{ChechikLPR12}
Shiri Chechik, Michael Langberg, David Peleg, and Liam Roditty.
\newblock f-sensitivity distance oracles and routing schemes.
\newblock {\em Algorithmica}, 63(4):861--882, 2012.

\bibitem[CT65]{fft}
James~W Cooley and John~W Tukey.
\newblock An algorithm for the machine calculation of complex fourier series.
\newblock {\em Mathematics of computation}, 19(90):297--301, 1965.

\bibitem[CWX21]{ChanWX21}
Timothy~M. Chan, Virginia~Vassilevska Williams, and Yinzhan Xu.
\newblock Algorithms, reductions and equivalences for small weight variants of
  all-pairs shortest paths.
\newblock In {\em 48th International Colloquium on Automata, Languages, and
  Programming, {ICALP} 2021}, volume 198 of {\em LIPIcs}, pages 47:1--47:21.
  Schloss Dagstuhl - Leibniz-Zentrum f{\"{u}}r Informatik, 2021.

\bibitem[CZ]{ChechikZ23}
Shiri Chechik and Tianyi Zhang.
\newblock {\em Faster Deterministic Worst-Case Fully Dynamic All-Pairs Shortest
  Paths via Decremental Hop-Restricted Shortest Paths}, pages 87--99.

\bibitem[CZ22]{ChechikZ22}
Shiri Chechik and Tianyi Zhang.
\newblock Nearly 2-approximate distance oracles in subquadratic time.
\newblock In {\em Proceedings of the 2022 {ACM-SIAM} Symposium on Discrete
  Algorithms, {SODA} 2022}, pages 551--580. {SIAM}, 2022.

\bibitem[DI04]{DemetrescuI04}
Camil Demetrescu and Giuseppe~F. Italiano.
\newblock A new approach to dynamic all pairs shortest paths.
\newblock {\em J. {ACM}}, 51(6):968--992, 2004.

\bibitem[DI05]{DemetrescuI05}
Camil Demetrescu and Giuseppe~F. Italiano.
\newblock Trade-offs for fully dynamic transitive closure on dags: breaking
  through the o(n\({}^{\mbox{2}}\) barrier.
\newblock {\em J. {ACM}}, 52(2):147--156, 2005.

\bibitem[DP09]{DuanP09a}
Ran Duan and Seth Pettie.
\newblock Dual-failure distance and connectivity oracles.
\newblock In {\em Proceedings of the Twentieth Annual {ACM-SIAM} Symposium on
  Discrete Algorithms, {SODA} 2009}, pages 506--515. {SIAM}, 2009.

\bibitem[DR22]{DuanR22}
Ran Duan and Hanlin Ren.
\newblock Maintaining exact distances under multiple edge failures.
\newblock In {\em {STOC} '22: 54th Annual {ACM} {SIGACT} Symposium on Theory of
  Computing}, pages 1093--1101. {ACM}, 2022.

\bibitem[DZ17]{DuanZ17a}
Ran Duan and Tianyi Zhang.
\newblock Improved distance sensitivity oracles via tree partitioning.
\newblock In {\em Algorithms and Data Structures - 15th International
  Symposium, {WADS} 2017}, volume 10389 of {\em Lecture Notes in Computer
  Science}, pages 349--360. Springer, 2017.

\bibitem[Ebe00]{eberly00}
Wayne Eberly.
\newblock Asymptotically efficient algorithms for the frobenius form.
\newblock Technical report, Department of Computer Science, Universiyt of
  Calgary, 2000.

\bibitem[EFGW21]{EvaldFGW21}
Jacob Evald, Viktor Fredslund{-}Hansen, Maximilian~Probst Gutenberg, and
  Christian Wulff{-}Nilsen.
\newblock Decremental {APSP} in unweighted digraphs versus an adaptive
  adversary.
\newblock In {\em 48th International Colloquium on Automata, Languages, and
  Programming, {ICALP} 2021}, volume 198 of {\em LIPIcs}, pages 64:1--64:20.
  Schloss Dagstuhl - Leibniz-Zentrum f{\"{u}}r Informatik, 2021.

\bibitem[FS11]{FrandsenS11}
Gudmund~Skovbjerg Frandsen and Piotr Sankowski.
\newblock Dynamic normal forms and dynamic characteristic polynomial.
\newblock {\em Theor. Comput. Sci.}, 412(16):1470--1483, 2011.

\bibitem[Gie95]{Giesbrecht95}
Mark Giesbrecht.
\newblock Nearly optimal algorithms for canonical matrix forms.
\newblock {\em {SIAM} J. Comput.}, 24(5):948--969, 1995.

\bibitem[GR21]{GuR21}
Yong Gu and Hanlin Ren.
\newblock Constructing a distance sensitivity oracle in {O}(n{\^{}}2.5794 {M)}
  time.
\newblock In {\em 48th International Colloquium on Automata, Languages, and
  Programming, {ICALP} 2021}, volume 198 of {\em LIPIcs}, pages 76:1--76:20.
  Schloss Dagstuhl - Leibniz-Zentrum f{\"{u}}r Informatik, 2021.

\bibitem[GU18]{GallU18}
Francois~Le Gall and Florent Urrutia.
\newblock Improved rectangular matrix multiplication using powers of the
  coppersmith-winograd tensor.
\newblock In {\em Proceedings of the Twenty-Ninth Annual {ACM-SIAM} Symposium
  on Discrete Algorithms, {SODA} 2018}, pages 1029--1046. {SIAM}, 2018.

\bibitem[GVL13]{golub2013matrix}
Gene~H Golub and Charles~F Van~Loan.
\newblock {\em Matrix computations}.
\newblock JHU press, 2013.

\bibitem[GW20a]{GV20}
Fabrizio Grandoni and Virginia~Vassilevska Williams.
\newblock Faster replacement paths and distance sensitivity oracles.
\newblock {\em {ACM} Trans. Algorithms}, 16(1):15:1--15:25, 2020.

\bibitem[GW20b]{GutenbergW20b}
Maximilian~Probst Gutenberg and Christian Wulff{-}Nilsen.
\newblock Fully-dynamic all-pairs shortest paths: Improved worst-case time and
  space bounds.
\newblock In {\em Proceedings of the 2020 {ACM-SIAM} Symposium on Discrete
  Algorithms, {SODA} 2020}, pages 2562--2574. {SIAM}, 2020.

\bibitem[Har09]{Harvey09}
Nicholas J.~A. Harvey.
\newblock Algebraic algorithms for matching and matroid problems.
\newblock {\em {SIAM} J. Comput.}, 39(2):679--702, 2009.

\bibitem[HP98]{HuangP98}
Xiaohan Huang and Victor~Y. Pan.
\newblock Fast rectangular matrix multiplication and applications.
\newblock {\em J. Complex.}, 14(2):257--299, 1998.

\bibitem[HS81]{SMW}
H.~V. Henderson and S.~R. Searle.
\newblock On deriving the inverse of a sum of matrices.
\newblock {\em SIAM Review}, 23(1):53--60, 1981.

\bibitem[JV05]{JeannerodV05}
Claude{-}Pierre Jeannerod and Gilles Villard.
\newblock Essentially optimal computation of the inverse of generic polynomial
  matrices.
\newblock {\em J. Complex.}, 21(1):72--86, 2005.

\bibitem[Kar21]{Karczmarz21}
Adam Karczmarz.
\newblock Fully dynamic algorithms for minimum weight cycle and related
  problems.
\newblock In {\em 48th International Colloquium on Automata, Languages, and
  Programming, {ICALP} 2021}, volume 198 of {\em LIPIcs}, pages 83:1--83:20.
  Schloss Dagstuhl - Leibniz-Zentrum f{\"{u}}r Informatik, 2021.

\bibitem[Kel85]{Keller-Gehrig85}
Walter Keller{-}Gehrig.
\newblock Fast algorithms for the characteristic polynomial.
\newblock {\em Theor. Comput. Sci.}, 36:309--317, 1985.

\bibitem[K{\L}19]{KarczmarzL19}
Adam Karczmarz and Jakub {\L}ącki.
\newblock Reliable hubs for partially-dynamic all-pairs shortest paths in
  directed graphs.
\newblock In {\em 27th Annual European Symposium on Algorithms, {ESA} 2019},
  volume 144 of {\em LIPIcs}, pages 65:1--65:15. Schloss Dagstuhl -
  Leibniz-Zentrum f{\"{u}}r Informatik, 2019.

\bibitem[KS02]{KingS02}
Valerie King and Garry Sagert.
\newblock A fully dynamic algorithm for maintaining the transitive closure.
\newblock {\em J. Comput. Syst. Sci.}, 65(1):150--167, 2002.

\bibitem[KS23]{KarczmarzS23}
Adam Karczmarz and Piotr Sankowski.
\newblock Fully dynamic shortest paths and reachability in sparse digraphs.
\newblock In {\em 50th International Colloquium on Automata, Languages, and
  Programming, {ICALP} 2023}, volume 261 of {\em LIPIcs}, pages 84:1--84:20.
  Schloss Dagstuhl - Leibniz-Zentrum f{\"{u}}r Informatik, 2023.

\bibitem[LPW20]{LincolnPW20}
Andrea Lincoln, Adam Polak, and Virginia~Vassilevska Williams.
\newblock Monochromatic triangles, intermediate matrix products, and
  convolutions.
\newblock In {\em 11th Innovations in Theoretical Computer Science Conference,
  {ITCS} 2020}, volume 151 of {\em LIPIcs}, pages 53:1--53:18. Schloss Dagstuhl
  - Leibniz-Zentrum f{\"{u}}r Informatik, 2020.

\bibitem[MS04]{MuchaS04}
Marcin Mucha and Piotr Sankowski.
\newblock Maximum matchings via gaussian elimination.
\newblock In {\em 45th Symposium on Foundations of Computer Science {(FOCS}
  2004)}, pages 248--255. {IEEE} Computer Society, 2004.

\bibitem[NP95]{NeumannP95}
Peter~M. Neumann and Cheryl~E. Praeger.
\newblock Cyclic matrices over finite fields.
\newblock {\em Journal of the London Mathematical Society}, 52(2):263--284,
  1995.

\bibitem[PR14]{PatrascuR14}
Mihai Patrascu and Liam Roditty.
\newblock Distance oracles beyond the thorup-zwick bound.
\newblock {\em {SIAM} J. Comput.}, 43(1):300--311, 2014.

\bibitem[Ren22]{Ren22}
Hanlin Ren.
\newblock Improved distance sensitivity oracles with subcubic preprocessing
  time.
\newblock {\em J. Comput. Syst. Sci.}, 123:159--170, 2022.

\bibitem[RZ11]{RodittyZ11}
Liam Roditty and Uri Zwick.
\newblock On dynamic shortest paths problems.
\newblock {\em Algorithmica}, 61(2):389--401, 2011.

\bibitem[San04]{Sankowski04}
Piotr Sankowski.
\newblock Dynamic transitive closure via dynamic matrix inverse (extended
  abstract).
\newblock In {\em 45th Symposium on Foundations of Computer Science, {FOCS}
  2004}, pages 509--517. {IEEE} Computer Society, 2004.

\bibitem[San05a]{Sankowski05a}
Piotr Sankowski.
\newblock Shortest paths in matrix multiplication time.
\newblock In {\em Algorithms - {ESA} 2005, 13th Annual European Symposium},
  volume 3669 of {\em Lecture Notes in Computer Science}, pages 770--778.
  Springer, 2005.

\bibitem[San05b]{Sankowski05}
Piotr Sankowski.
\newblock Subquadratic algorithm for dynamic shortest distances.
\newblock In {\em Computing and Combinatorics, 11th Annual International
  Conference, {COCOON} 2005}, volume 3595 of {\em Lecture Notes in Computer
  Science}, pages 461--470. Springer, 2005.

\bibitem[San07]{Sankowski07}
Piotr Sankowski.
\newblock Faster dynamic matchings and vertex connectivity.
\newblock In {\em Proceedings of the Eighteenth Annual {ACM-SIAM} Symposium on
  Discrete Algorithms, {SODA} 2007}, pages 118--126. {SIAM}, 2007.

\bibitem[Sch80]{Schwartz80}
Jacob~T. Schwartz.
\newblock Fast probabilistic algorithms for verification of polynomial
  identities.
\newblock {\em J. {ACM}}, 27(4):701--717, 1980.

\bibitem[Som14]{Sommer14}
Christian Sommer.
\newblock Shortest-path queries in static networks.
\newblock {\em {ACM} Comput. Surv.}, 46(4):45:1--45:31, 2014.

\bibitem[Sto01]{Storjohann01}
Arne Storjohann.
\newblock Deterministic computation of the frobenius form.
\newblock In {\em 42nd Annual Symposium on Foundations of Computer Science,
  {FOCS} 2001}, pages 368--377. {IEEE} Computer Society, 2001.

\bibitem[Sto15]{Storjohann15}
Arne Storjohann.
\newblock On the complexity of inverting integer and polynomial matrices.
\newblock {\em Comput. Complex.}, 24(4):777--821, 2015.

\bibitem[SW19]{SankowskiW19}
Piotr Sankowski and Karol Wegrzycki.
\newblock Improved distance queries and cycle counting by frobenius normal
  form.
\newblock {\em Theory Comput. Syst.}, 63(5):1049--1067, 2019.

\bibitem[Tho04]{Thorup04}
Mikkel Thorup.
\newblock Fully-dynamic all-pairs shortest paths: Faster and allowing negative
  cycles.
\newblock In {\em {SWAT} 2004, 9th Scandinavian Workshop on Algorithm Theory},
  volume 3111 of {\em Lecture Notes in Computer Science}, pages 384--396.
  Springer, 2004.

\bibitem[Tho05]{Thorup05}
Mikkel Thorup.
\newblock Worst-case update times for fully-dynamic all-pairs shortest paths.
\newblock In {\em Proceedings of the 37th Annual {ACM} Symposium on Theory of
  Computing 2005}, pages 112--119. {ACM}, 2005.

\bibitem[TZ05]{ThorupZ05}
Mikkel Thorup and Uri Zwick.
\newblock Approximate distance oracles.
\newblock {\em J. {ACM}}, 52(1):1--24, 2005.

\bibitem[UY91]{UY91}
Jeffrey~D. Ullman and Mihalis Yannakakis.
\newblock High-probability parallel transitive-closure algorithms.
\newblock {\em {SIAM} J. Comput.}, 20(1):100--125, 1991.

\bibitem[vdB21]{Brand21}
Jan van~den Brand.
\newblock Unifying matrix data structures: Simplifying and speeding up
  iterative algorithms.
\newblock In {\em 4th Symposium on Simplicity in Algorithms, {SOSA} 2021},
  pages 1--13. {SIAM}, 2021.

\bibitem[vdBFN21]{dynamic-st-dist}
Jan van~den Brand, Sebastian Forster, and Yasamin Nazari.
\newblock Fast deterministic fully dynamic distance approximation.
\newblock {\em CoRR}, abs/2111.03361, 2021.

\bibitem[vdBFN22]{BrandFN22}
Jan van~den Brand, Sebastian Forster, and Yasamin Nazari.
\newblock Fast deterministic fully dynamic distance approximation.
\newblock In {\em 63rd {IEEE} Annual Symposium on Foundations of Computer
  Science, {FOCS} 2022}, pages 1011--1022. {IEEE}, 2022.

\bibitem[vdBN19]{BrandN19}
Jan van~den Brand and Danupon Nanongkai.
\newblock Dynamic approximate shortest paths and beyond: Subquadratic and
  worst-case update time.
\newblock In {\em 60th {IEEE} Annual Symposium on Foundations of Computer
  Science, {FOCS} 2019}, pages 436--455. {IEEE} Computer Society, 2019.

\bibitem[vdBNS19]{BrandNS19}
Jan van~den Brand, Danupon Nanongkai, and Thatchaphol Saranurak.
\newblock Dynamic matrix inverse: Improved algorithms and matching conditional
  lower bounds.
\newblock In {\em 60th {IEEE} Annual Symposium on Foundations of Computer
  Science, {FOCS} 2019}, pages 456--480. {IEEE} Computer Society, 2019.

\bibitem[vdBS19]{BrandS19}
Jan van~den Brand and Thatchaphol Saranurak.
\newblock Sensitive distance and reachability oracles for large batch updates.
\newblock In {\em 60th {IEEE} Annual Symposium on Foundations of Computer
  Science, {FOCS} 2019}, pages 424--435. {IEEE} Computer Society, 2019.

\bibitem[Vil00]{villard00}
Gilles Villard.
\newblock Computing the frobenius normal form of a sparse matrix.
\newblock In {\em Computer Algebra in Scientific Computing}, pages 395--407,
  Berlin, Heidelberg, 2000. Springer Berlin Heidelberg.

\bibitem[Wie86]{Wiedemann86}
Douglas~H. Wiedemann.
\newblock Solving sparse linear equations over finite fields.
\newblock {\em {IEEE} Trans. Inf. Theory}, 32(1):54--62, 1986.

\bibitem[Wul12]{Wulff-Nilsen12}
Christian Wulff{-}Nilsen.
\newblock Approximate distance oracles with improved preprocessing time.
\newblock In {\em Proceedings of the Twenty-Third Annual {ACM-SIAM} Symposium
  on Discrete Algorithms, {SODA} 2012}, pages 202--208. {SIAM}, 2012.

\bibitem[WY13]{WeimannY13}
Oren Weimann and Raphael Yuster.
\newblock Replacement paths and distance sensitivity oracles via fast matrix
  multiplication.
\newblock {\em {ACM} Trans. Algorithms}, 9(2):14:1--14:13, 2013.

\bibitem[YZ05]{YusterZ05}
Raphael Yuster and Uri Zwick.
\newblock Answering distance queries in directed graphs using fast matrix
  multiplication.
\newblock In {\em 46th Annual {IEEE} Symposium on Foundations of Computer
  Science {(FOCS} 2005)}, pages 389--396. {IEEE} Computer Society, 2005.

\bibitem[Zip79]{Zippel79}
Richard Zippel.
\newblock Probabilistic algorithms for sparse polynomials.
\newblock In {\em Symbolic and Algebraic Computation, {EUROSAM} '79, An
  International Symposiumon Symbolic and Algebraic Computation}, volume~72 of
  {\em Lecture Notes in Computer Science}, pages 216--226. Springer, 1979.

\bibitem[ZLS15]{ZhouLS15}
Wei Zhou, George Labahn, and Arne Storjohann.
\newblock A deterministic algorithm for inverting a polynomial matrix.
\newblock {\em J. Complex.}, 31(2):162--173, 2015.

\bibitem[Zwi02]{Zwick02}
Uri Zwick.
\newblock All pairs shortest paths using bridging sets and rectangular matrix
  multiplication.
\newblock {\em J. {ACM}}, 49(3):289--317, 2002.

\end{thebibliography}

\end{document}